\documentclass[11pt,reqnoa,4paper]{amsart}

\usepackage{amscd,amssymb,amsmath,amsthm}
\usepackage{amsmath,amsfonts,amssymb,amsthm, bm}
\usepackage{graphicx,epstopdf}
\usepackage{color}
\usepackage{cite}
\usepackage{fullpage}
\usepackage{mathrsfs} %for nicer caligraphic letters via mathscr
\usepackage[dvipsnames]{xcolor}
\usepackage[pagebackref=true,colorlinks=true, linkcolor=RoyalBlue, citecolor=RoyalBlue]{hyperref}

\newtheorem{thm}{Theorem}

\newtheorem{lemma}{Lemma}
\newtheorem{pro}{Proposition}
\newtheorem{rk}{Remark}
\newtheorem{re}{Result}

\numberwithin{equation}{section} 

\begin{document}
\title[Gibbs measures for the $p$-SOS model]{
Three-state $p$-SOS models on binary Cayley trees}

\author{Benedikt Jahnel \& Utkir Rozikov}

\address{B.\ Jahnel\\ Institut f\"ur Mathematische Stochastik,
	 Technische Universit\"at Braunschweig, Universit\"atsplatz 2, 38106, Braunschweig, Germany; \&
	Weierstrass Institute for Applied Analysis and Stochastics, Mohrenstrasse 39, 10117, Berlin, Germany.}
\email {benedikt.jahnel@tu-braunschweig.de}

	\address{U.\ Rozikov\\ V.I.Romanovskiy Institute of Mathematics, Uzbekistan Academy of Sciences, 9, Universitet str., 100174, Tashkent, Uzbekistan; \&  National University of Uzbekistan,  4, Universitet str., 100174, Tashkent, Uzbekistan;\&   Karshi State University, 17, Kuchabag str., 180119, Karshi, Uzbekistan.}
\email{rozikovu@yandex.ru}

\begin{abstract}
We consider a version of the solid-on-solid model on the Cayley tree of order two in which vertices carry spins of value $0,1$ or $2$ and the pairwise interaction of neighboring vertices is given by their spin difference to the power $p>0$. We exhibit all translation-invariant splitting Gibbs measures (TISGMs) of the model and demonstrate the existence of up to seven such measures, depending on the parameters.
We further establish general conditions for extremality and  non-extremality of TISGMs in the set of all Gibbs measures and use them to examine selected TISGMs for a small and a large $p$.
Notably, our analysis reveals that extremality properties are similar for large $p$ compared to the case $p=1$, a case that has been explored already in previous work. However, for the small $p$, certain measures that were consistently non-extremal for $p=1$ do exhibit transitions between extremality and non-extremality.
\end{abstract}
\maketitle



\section{Introduction}
In this paper we consider spin-configurations $\sigma$ which are functions from the vertices of a Cayley tree of order $ k\geq 1$ (that is an infinite graph without cycles such that exactly $k+1$ edges originate from each vertex) to the local configuration space $\Phi=\{0,1,\dots,m\}$, $m\geq 1$. For most of our analysis we will restrict to the case $m=2$ and the Cayley tree with $k=2$.  
By $\langle x,y\rangle$ we denote a pair of nearest-neighbor vertices. A two-parametric solid-on-solid model (called $p$-SOS) is a spin system with spins taking values in $\Phi$, and with formal Hamiltonian
\begin{align}\label{def_pSOS}
 H(\sigma)=-J\sum_{\langle x,y\rangle}
|\sigma(x)-\sigma(y)|^p,
\end{align}
where $p>0$ and $J\in \mathbb{R}$ the coupling constant.

\medskip
Let us note that, if $m=1$, i.e., $\Phi=\{0,1\}$ then, the $p$-SOS model can be reduced to the classical Ising (i.e., $2$-state Potts) model, since in this case, $|\sigma(x)-\sigma(y)|^p\in \{0,1\}$ for any $p>0$ and $\langle x,y\rangle$. A complete analysis of Gibbs measures for tree-indexed Ising and Potts models can be found for example in~\cite{Ro,RoP}. 
On the other hand, for $m\geq 2$ and $p=1$ the model becomes the classical SOS-model which is considered on the cubic lattice in~\cite{BK1,Maz} and on Cayley trees in~\cite{KR-sos,Ro12,Ro13,Ro}.
The case $p=2$ is known as the discrete Gaussian case, see for example~\cite{BP,Ve} and references therein.
In our recent paper \cite{JR1}, we investigated the case where $p = \infty$. For three states, $k = 2, 3$, and with increasing coupling strength, we demonstrated that the number of translation-invariant Gibbs measures evolves as follows:  $1\to3\to5\to6\to7$. 

%SOS models with $\Phi=\Z$ have been used as simplified discrete
%interface models which should approximate the behavior of a Dobrushin-state
%in an Ising model when the underlying graph is $\Z^d$, and $d\geq 2$.
%While there is the issue of possible
%non-existence of Gibbs-states in the case of such unbounded spins, in particular in the additional presence of {\em  disorder} (see \cite{BK} and \cite{BK1}),  this issue is not present here, and we are looking for a classification of the phases.
%
%Indeed, compared to the Potts model, the SOS model has
%less symmetry: The full symmetry of the Hamiltonian
%under joint permutation of the spin values is reduced to the mirror symmetry,
%which is the invariance of the model under the map $\s_i \mapsto m- \s_i$ on the local spin space.
%Therefore one expects a more diverse structure
%of phases. Note that, in the ferromagnetic  case
%it is intuitively plausible that the ground states
%corresponding to `middle-level
%surfaces' will be `dominant' as they carry more entropy.
% This observation was made formal
%in \cite{Maz} for the model on a cubic lattice.
%
%To the best of our knowledge, the first paper devoted to the SOS model on the Cayley tree is \cite{Ro12}.
In~\cite{Ro12} the SOS model (for $p=1$) is treated and a vector-valued functional equation for possible boundary laws of the model is obtained. It is known that each solution to this functional equation determines a splitting Gibbs measure (SGM) and in particular, the vertex-independent boundary laws define translation-invariant (TI) SGMs.

%See also \cite{Bl}, \cite{Ro} and the bibliography therein
%for more details about SOS and other models on trees. A natural direction for future work on Cayley tree can be
%devoted to the Widom-Rowlinson model as in \cite{MSS}.

\medskip
In this paper, we take a similar approach and present a description of all TISGMs for the three-state ($m=2$) $p$-SOS model on the Cayley tree of order two via solutions for the fixed-point equations of the vertex-independent boundary laws. This is a non-trivial extension of the analysis presented in~\cite{KR-sos} which deals only with the case $p=1$.

\medskip
In the recent paper~\cite{CKL} the authors consider $p$-SOS models with spin values in $\mathbb Z$ on Cayley trees of large degrees $k$. There, a family of extremal tree-automorphism non-invariant Gibbs measures is presented that arises
as low temperature perturbations of ground states. Moreover, the extremality of low-temperature states in the set of all Gibbs measures is shown. Our considerations on the
binary tree provide an extension of what was previously known for $p$-SOS model.

\medskip
The paper is organized as follows. In Section~\ref{sec_2} we introduce the general setup and present the defining functional equations for the $p$-SOS model. In Section~\ref{sec_3} we present the description of all TISGMs and show that their number can be up to seven, for any given parameters $p>0$ and $\theta=\exp(J)$. Finally, in Section~\ref{sec_5} we study the extremality questions for TISGMs and use the methods of~\cite{KR,KR-sos} based on the Kesten--Stigum's non-extremality condition \cite{Ke} and the Martinelli--Sinclair--Weitz's extremality condition, see~\cite{MSW}.

\section{Setup and functional equations}\label{sec_2}

We denote by $\Gamma^k=(V, L)$ the {\em Cayley tree} of order $k\geq 1$, where $V$ is the set of {\em vertices} and  $L$ the set of {\em edges}.
%Two vertices $x$ and $y$ are called {\it nearest neighbors} if there exists an
%edge $l \in L$ connecting them.
%We will use the notation $l=\langle x,y\rangle$.
A collection of nearest-neighbor pairs of vertices $\langle x,x_1\rangle, \langle x_1,x_2\rangle,...,\langle x_{d-1},y\rangle$ is called a {\em
path} from $x$ to $y$. The distance $d(x,y)$ on the Cayley tree is the number of edges of the shortest path from $x$ to $y$.
For a fixed $x^0\in V$, called the {\em root}, we set
\begin{equation*}
W_n:=\{x\in V\,| \, d(x,x^0)=n\}, \qquad V_n:=\bigcup_{m=0}^n W_m
\end{equation*}
and denote by
$$
S(x):=\{y\in W_{n+1} :  d(x,y)=1 \}, \ \ x\in W_n, $$ the set  of {\em direct successors} of $x$.
%{\it $p$-SOS model.} We consider models where the spin takes values in the set
%$\Phi:=\{0,1,\ldots , m\}$, $m\geq 2$, and is assigned to the vertices
%of the tree. A configuration $\sigma$ on $V$ is then defined
%as a function $x\in V\mapsto\sigma (x)\in\Phi$;
%the set of all configurations is $\Phi^V$.
%The (formal) Hamiltonian is of an $p$-SOS form:
%\begin{equation}\label{rs1.1}
% H(\sigma)=-J\sum_{\langle x,y\rangle\in L}
%|\sigma(x)-\sigma(y)|^p,
%\end{equation}
%where $p>0$ and $J\in \mathbb{R}$ is a coupling constant.
%
%Here, $J<0$
%gives a ferromagnetic  and $J>0$ an anti-ferromagnetic
%model.
%We use a standard definition  of a Gibbs measure (which is an infinite-volume measure
%which satisfies the DLR equation), and of a translation-invariant (TI)
%measure (which is a measure which is invariant under translations which map the
%tree onto itself). Also, we call measure $\mu$ {\em symmetric} 
%if it is preserved under the simultaneous change $j\mapsto m-j$
%at each vertex $x\in V$.
Next, we consider real vector-valued function from $V\setminus
\{x^0\}$ to $\mathbb{R}^{m+1}$ given as
$$h\colon x\mapsto h_x=(h_{0,x}, h_{1,x},\dots,h_{m,x}),$$ 
and the corresponding probability distributions $\mu^{(n)}$ on
$\Phi^{V_n}$, the set of all configuration given on $V_n$, $n\in \mathbb N$, defined as
\begin{equation}\label{rs2.1}
\mu^{(n)}(\sigma_n):=Z_n^{-1}\exp\Big(-H(\sigma_n)
+\sum_{x\in W_n}h_{\sigma(x),x}\Big).
\end{equation}
Here, $\sigma_n\colon x\in V_n\mapsto \sigma(x)\in \{0,\dots,m\}$
and $Z_n$ is the corresponding partition function \begin{equation}\label{rs2.2}
Z_n:=\sum_{{\widetilde\sigma}_n\in\Phi^{V_n}}
\exp\Big(-H({\widetilde\sigma}_n)
+\sum_{x\in W_n}h_{{\widetilde\sigma}(x),x}\Big).
\end{equation}
We say that the sequence of probability distributions $(\mu^{(n)})_{n\ge 1}$
are {\em compatible} if for all $n\geq 1$ and $\sigma_{n-1}\in\Phi^{V_{n-1}}$ we have that
\begin{equation}\label{rs2.3}
\sum_{\omega_n\in\Phi^{W_n}}\mu^{(n)}(\sigma_{n-1}\vee\omega_n)=
\mu^{(n-1)}(\sigma_{n-1}).
\end{equation}
Here $\sigma_{n-1}\vee\omega_n\in\Phi^{V_n}$ is the concatenation
of $\sigma_{n-1}$ and $\omega_n$.
If this is the case, then, by the Kolmogorov extension theorem, there exists a unique measure $\mu$ on
$\Phi^V$ such that, for all $n$ and
$\sigma_n\in\Phi^{V_n}$, 
$$\mu(\{\sigma
|_{V_n}=\sigma_n\})=\mu^{(n)}(\sigma_n).$$ 
Such
a measure is called a {\em splitting Gibbs measure} (SGM) corresponding to the Hamiltonian $H$ and function $x\mapsto h_x$, $x\neq x^0$.
Let us note that these measures are also called {\em Markov chains} for example in~\cite{Pr}.

\medskip
Let us now turn our attention to the $p$-SOS model with Hamiltonian defined in~\eqref{def_pSOS}. The following statement presents a system of functional equations whose solutions correspond to (infinite-volume) Gibbs measures of the $p$-SOS model on Cayley trees. Let us note that every extremal Gibbs measure also arises
in this way, but not necessarily every measure which arises in this way is extremal,
see~\cite[Chapter 12]{Ge}.
%We recall the derivation of the equations via the compatibility requirement
%for the convenience of the reader.
In other words, the following statement describes conditions on $h_x$ guaranteeing compatibility of distributions $\mu^{(n)}(\sigma_n)$.
\begin{pro}\label{rsp2.1} Probability distributions
$\mu^{(n)}$, $n=1,2,\ldots$, in~\eqref{rs2.1} are compatible iff, for any $x\in V\setminus\{x^0\}$,
the following equation holds,
\begin{equation}\label{rs2.4}
 h^*_x=\sum_{y\in S(x)}F(h^*_y,m,\theta).
 \end{equation}
Here $\theta:=\exp(J)$ and $h^*_x$ stands for the vector
$(h_{0,x}-h_{m,x}, h_{1,x}-h_{m,x},\dots,h_{m-1,x}-h_{m,x})$ and
the vector-valued function $F(\cdot,m,\theta )\colon\mathbb{R}^m\to \mathbb{R}^m$
is $F(h,m,\theta ):=
(F_0(h,m,\theta ),\dots,F_{m-1}(h,m,\theta))$, with
\begin{equation}\label{rs2.6}
F_i(h,m,\theta ):=\ln{\sum_{j=0}^{m-1}
\theta^{|i-j|^p}\exp(h_j)+\theta^{(m-i)^p}\over
\sum_{j=0}^{m-1}\theta^{(m-j)^p}\exp(h_j)+1},
\end{equation}
where $h:=(h_0,h_1,\dots,h_{m-1}),
i=0,\dots ,m-1.$
\end{pro}
\begin{proof} The proof is similar to the proof of~\cite[Proposition 2.1.]{Ro12}. 
	\end{proof}
From Proposition~\ref{rsp2.1} it follows that for any $h=\{h_x\colon x\in V\}$
satisfying~\eqref{rs2.4} there exists a unique SGM $\mu$ for the $p$-SOS model. However,
the analysis of solutions to~\eqref{rs2.4} for an arbitrary $m$ is challenging. We therefore restrict our attention to a smaller class of measures, namely the translation-invariant SGMs.

\medskip
It is natural to begin with translation-invariant solutions
where $h_x=h\in \mathbb{R}^m$ is independent of $x$. In this case~\eqref{rs2.4} becomes
\begin{equation}\label{ti0}
z_i=\left(\sum_{j=0}^{m-1}
\theta^{|i-j|^p}z_j+\theta^{(m-i)^p}\over
\sum_{j=0}^{m-1}\theta^{(m-j)^p}z_j+1\right)^k,\qquad i=0,\dots, m-1,
\end{equation}
where $z_i=\exp(h_i)$. The vector $(z_0,\dots, z_{m-1})$ is called a (translation-invariant) {\em law} and the non-translation-invariant quantities $l_x(i)=\exp(h_{i,x})$
are the {\em boundary laws}, see~ \cite{Za} and \cite[pp.~242]{Ge}.
%\begin{rk} The system of equations~\eqref{ti0}) has parameters $k\geq 2$, $m\geq 2$ and $\theta>0$,
%and it seems very difficult to find all solutions in the general case.
%
%In cases $m=2$ and $m=3$ the existence of mirror symmetric solutions (i.e. with  $z_{m-j}=z_j$, $j=0,1,\dots,m$) to the system~\eqref{ti0}) were studied in \cite{Ro12} and \cite{Ro13}.
In the present manuscript we present a full analysis of
the solutions of the system~\eqref{ti0} for the case where $k=2$, $p>0$, $m=2$ and additionally study extremality properties of the corresponding TISGMs.  
In particular, we extend the results in~\cite{KR-sos} for general $p>0$, which were obtained there only for $p=1$. 
\section{The case $k=m=2$: complete analysis of solutions}\label{sec_3}

Assuming $k=m=2$, the two-dimensional fixed-point
equation~\eqref{ti0} for the two components of the boundary law can be written in terms of the variables $x=\sqrt{z_0}$ and $y=\sqrt{z_1}$ in the form
\begin{equation}\label{rs3.2a}
x={x^2+\theta y^2+\theta^{2^p} \over \theta^{2^p}x^2+\theta y^2+1},
\end{equation}
\begin{equation}\label{rs3.2b}
 y={\theta x^2+y^2+\theta \over \theta^{2^p}x^2+\theta y^2+1}.
\end{equation}
From~\eqref{rs3.2a} we get $x=1$ or
\begin{equation}\label{y}
\theta y^2=(1-\theta^{2^p})x-\theta^{2^p}(x^2+1).
\end{equation}
\begin{rk}\label{<1} 
Since $x>0$ we have that~\eqref{y} can hold iff $\theta<1$.
\end{rk}
\subsection{Case: $x=1$.} In this case, from~\eqref{rs3.2b}, we get
\begin{equation}\label{y3}
\theta y^3-y^2+(\theta^{2^p}+1)y-2\theta=0.
\end{equation}
Using Cardano's formula, the discriminant of the cubic equation~\eqref{y3} can be written as 
\begin{equation}\label{De}
	\Delta:= \Delta(\theta,p):={1\over 27\theta^2}\left[ 4(1-3\theta-3\theta^{2^p+1})^3-(2-9\theta+54\theta^3-9\theta^{2^p+1})^2\right].
	\end{equation}
From this we get the following statement.
\begin{lemma}\label{l1} The following assertions hold
\begin{itemize}
\item If $\Delta>0$, then~\eqref{y3} has three solutions $0<y_3<y_2<y_1$.
\item If $\Delta=0$, then~\eqref{y3} has two solutions $0<y_2<y_1$.
\item If $\Delta<0$, then~\eqref{y3} has one solution $0<y_1$.
\end{itemize}
\end{lemma}
\begin{proof} It is well known that the number of real roots for a cubic equation is determined by the value of $\Delta$. It remains to show that all real roots are positive. Descartes' rule of signs is very helpful to count the number of the positive roots of a polynomial. In~\eqref{y3} the sign of the coefficients changes three times, and hence, there might be one positive root or three positive roots. If $\Delta<0$,
then the unique real root is positive. If $\Delta\geq 0$,
then the rule of signs applied to the cubic equation, with $-y$
as a substitute, tells us that there are no negative roots, and since we can not have $\theta=0$, there must be three positive roots.
\end{proof}

\subsection{Case: $x\ne 1$ and~\eqref{y} satisfied.} By Remark~\ref{<1}, we only consider the case $\theta<1$ and
\eqref{rs3.2b} can be written as
\begin{equation}\label{2b}
 y^2=\left({\theta x^2+y^2+\theta \over \theta^{2^p}x^2+\theta y^2+1}\right)^2.
\end{equation}
In the case, when~\eqref{y} is satisfied,
from~\eqref{2b} we get that
\begin{equation}\label{x4}((1-\theta^{2^p})x-\theta^{2^p}(x^2+1))\theta=\left({(\theta^2-\theta^{2^p})(x^2+1)+(1-\theta^{2^p})x\over (1-\theta^{2^p})(x+1)}\right)^2.
\end{equation}
Let $\xi:=x+1/x$ and note that $\xi>2$ if $x>0$. Then, from~\eqref{x4} we get
\begin{equation}\label{xi}
a\xi^2+b\xi+c=0,
\end{equation}
where 
$$\begin{array}{lll}
a:=a(\theta, p):=\theta^{2^p+1}(1-\theta^{2^p})^2+(\theta^2-\theta^{2^p})^2,\\[2mm]
b:=b(\theta,p):=(1-\theta^{2^p})[2(\theta^2-\theta^{2^p})-\theta(1-\theta^{2^p})(1-3\theta^{2^p})],\\[2mm]
c:=c(\theta, p):=(1-\theta^{2^p})^2[1-2\theta(1-\theta^{2^p})].
\end{array} $$
This equation has no solution if $D=b^2-4ac<0$, it has a
unique solution if $D=0$ and two solutions if $D>0$.
Thus we have the following statements.
\begin{itemize}
	\item If $D>0$, then~\eqref{xi} has two solutions $\xi_1(\theta,p)<\xi_2(\theta,p)$ (below we denote $q=1-\theta^{2^p}$) given by
	\begin{equation}\label{12} 
 \xi_{1,2}(\theta,p):={q\over 2}{-3\theta q^2+2(\theta+1)q+2(\theta^2-1)\mp\theta\sqrt{q(q+2\theta-2)[(q-\theta-1)^2+(\theta+1)(3\theta-1)]}\over (q-\theta-1)[\theta q^2+(\theta^2-1)(q+\theta-1)]}.
		\end{equation}
	\item If $D=0$, then~\eqref{xi} has a unique solution
 $$\xi_{1,2}:={q\over 2}{-3\theta q^2+2(\theta+1)q+2(\theta^2-1)\over (q-\theta-1)[\theta q^2+(\theta^2-1)(q+\theta-1)]}.$$
	\item If $D<0$, then~\eqref{xi} has no solution.
\end{itemize}
For $D=0$ we have 
\begin{equation}
	\label{r}
	D(\theta,p)=\theta^2(\theta^{2^p}-1)^3(\theta^{2^p}-2\theta+1)
	((\theta^{2^p}+\theta)^2+3\theta^2+2\theta-1)=0.
\end{equation}
Define 
\begin{align*}
l(\theta, p):=\theta^{2^p}-2\theta+1\qquad\text{and}\qquad
q(\theta, p):=(\theta^{2^p}+\theta)^2+3\theta^2+2\theta-1.
\end{align*}
Since $0<\theta<1$ we see that $D>0$ if (see Fig.~\ref{ey5})
\begin{equation}\label{lq}
	l(\theta,p)q(\theta,p)<0.
\end{equation}
\begin{figure}[h]
	\begin{center}
		\includegraphics[width=8.3cm]{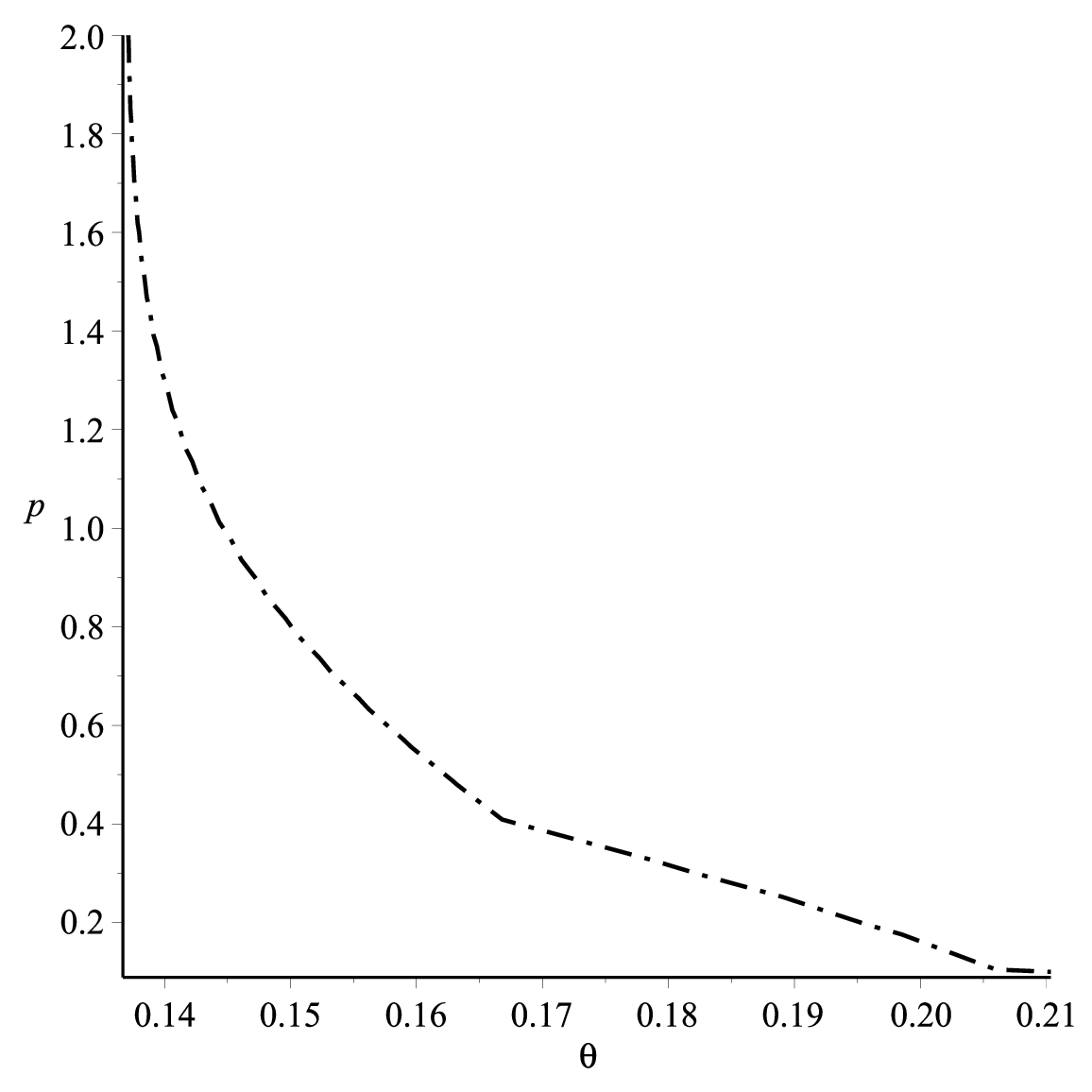}
			\includegraphics[width=5cm]{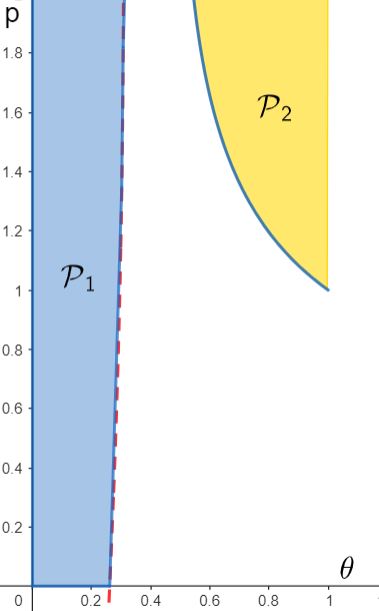}
	\end{center}
	\caption{{\it Left}: The dash-dot curve defined by $\Delta=0$ with~\eqref{De}. {\it Right}: The areas of $\mathcal P_1$ and $\mathcal P_2$, where the dashed (red) curve is $q(\theta, p)=0$ and the boundary (blue) curve of $\mathcal P_2$ is $l(\theta,p)=0$ defined by~\eqref{r}.}\label{ey5}
\end{figure}

Next, define 
\begin{align*}
m(\theta)&={1\over \ln2}\ln\left({\ln(-\theta+\sqrt{(\theta+1)(1-3\theta)})\over \ln \theta}\right), \ \ \theta\in \left(0, {\sqrt{5}-1\over 4}\right),\text{ and}\\
M(\theta)&={1\over \ln2}\ln\left({\ln(2\theta-1)\over \ln \theta}\right), \ \ \theta\in \left({1\over 2}, 1\right).
\end{align*}
We solve~\eqref{lq} with respect to $p$ and get the following solution,
\begin{equation}
	\mathcal P:=\mathcal P_1\cup \mathcal P_2:=\left\{(\theta, p) : 0<\theta<\tfrac{\sqrt{5}-1}{4}, p>m(\theta)\right\}
	\bigcup \left\{(\theta, p) : \tfrac{1}{2}<\theta<1, p>M(\theta)\right\}.
\end{equation}
Introduce (see Lemma \ref{l1}, the case $x=1$):
$$\mathcal Q_-=\left\{(\theta, p) : \Delta(\theta, p)<0\right\},\quad\mathcal Q_0=\left\{(\theta, p) : \Delta(\theta, p)=0\right\}\quad\text{and}\quad\mathcal Q_+=\left\{(\theta, p) : \Delta(\theta, p)>0\right\},$$
and note that $\mathcal Q_+\subset \mathcal P_1$.
Next, for all $\theta, p$ with $D\geq 0$ and
\begin{equation}\label{C1}
	2<\xi_1(\theta, p)\leq \xi_2(\theta, p)
\end{equation} we find all four positive solutions to~\eqref{x4} explicitly, i.e., we have
\begin{equation}\label{x4-7}
\begin{array}{ll}
		x_{4}(\theta, p):={1\over 2}(\xi_2-\sqrt{\xi_2^2-4}), \ \ x_{5}(\theta, p)={1\over 2}(\xi_1-\sqrt{\xi_1^2-4}),\\[3mm]
		x_{6}(\theta, p):={1\over 2}(\xi_1+\sqrt{\xi_1^2-4}),\ \ x_{7}(\theta, p)={1\over 2}(\xi_2+\sqrt{\xi_2^2-4}).
	\end{array}
\end{equation}
Fig.~\ref{four} presents the graphs of $x_i$, $i=4,5,6,7$.
\begin{figure}[h!]
	\begin{center}
		\includegraphics[width=9cm]{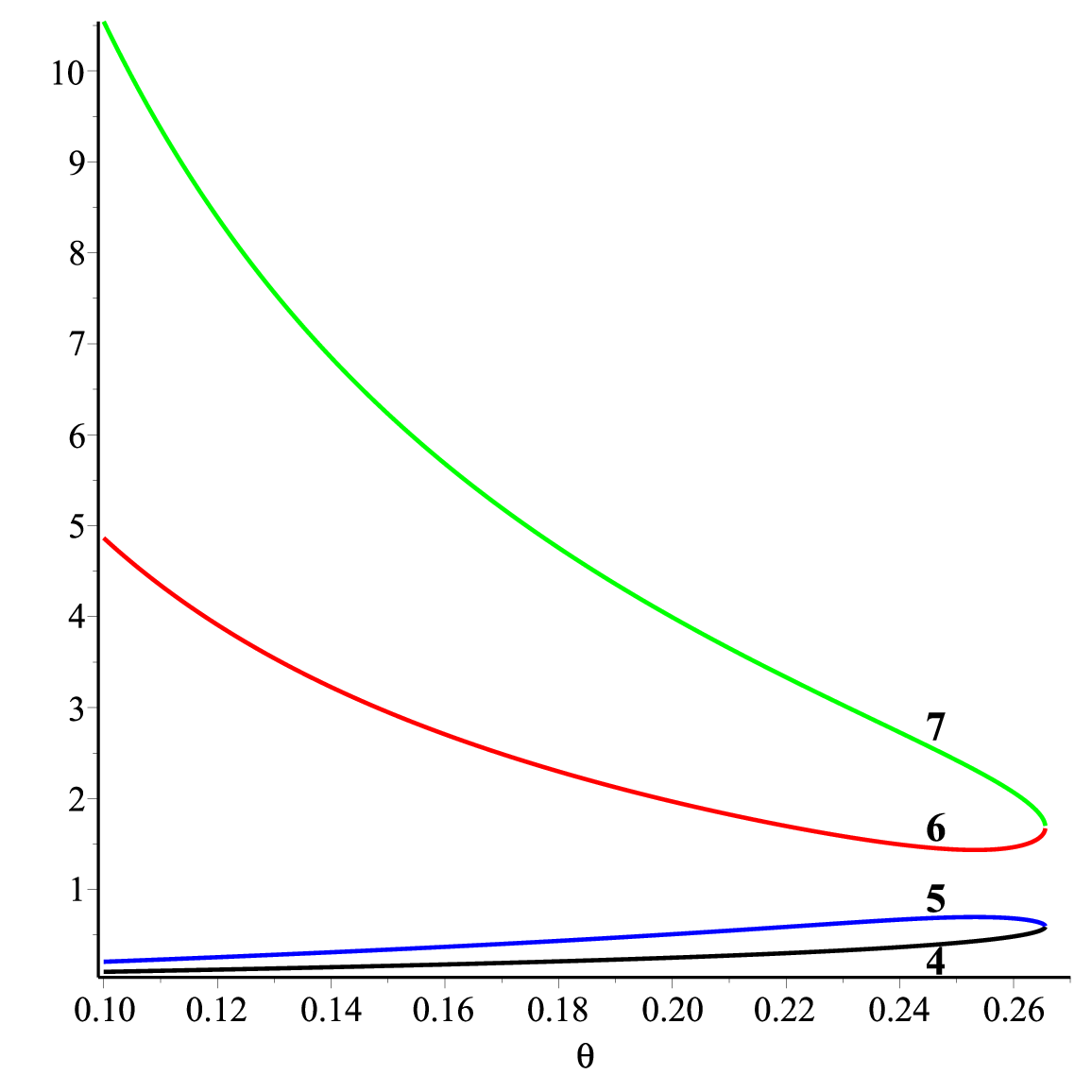}
	\end{center}
	\caption{ The graphs of the functions $x_i=x_i(\theta, 0.1)$, $i=4,5,6,7$. The coloring represents $x_4=black$, $x_5=blue$, $x_6=red$ and $x_7=green$.}\label{four}
%\end{figure}
%	\begin{figure}[h]
	\begin{center}
		\includegraphics[width=9cm]{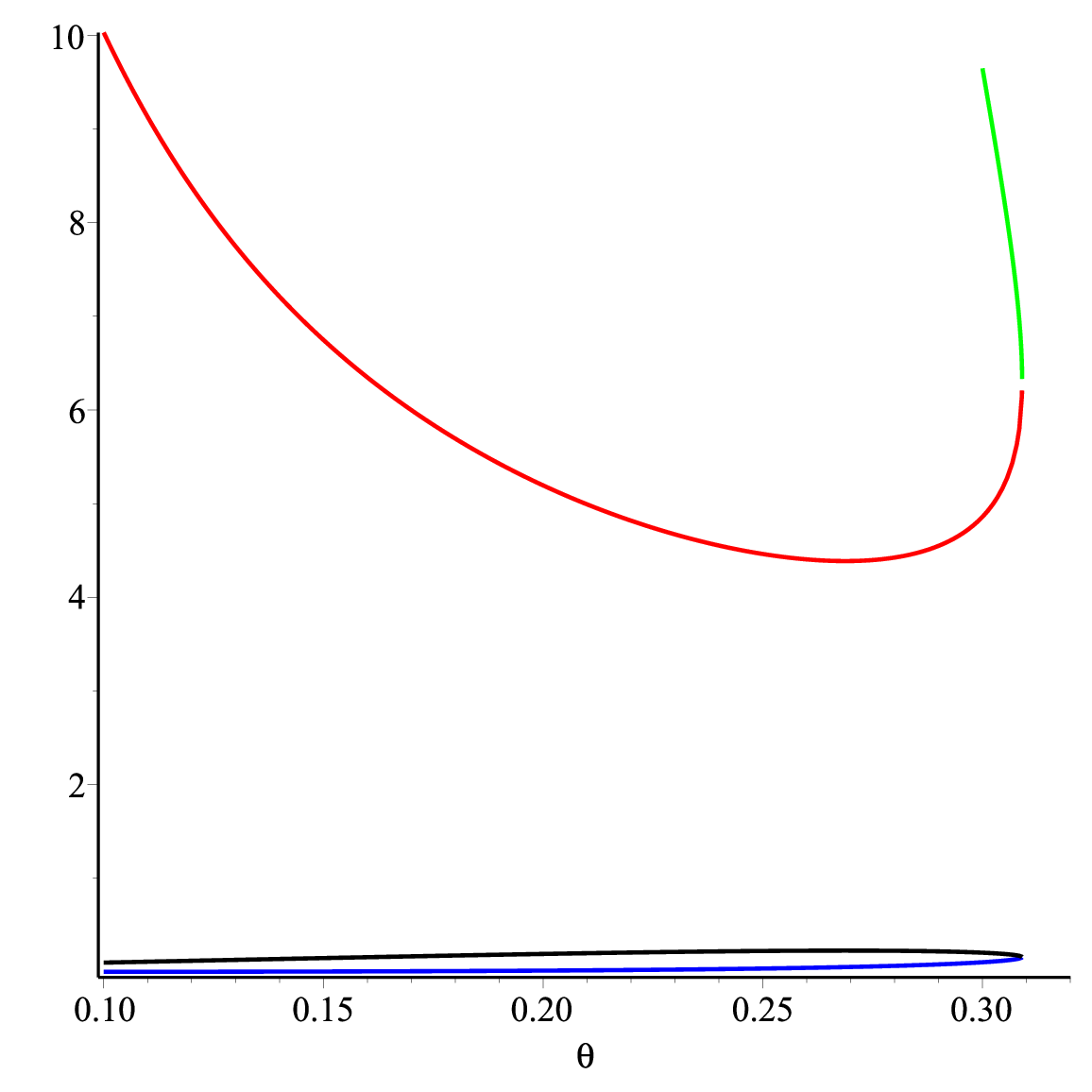}
	\end{center}
	\caption{ The graphs of the functions $x_i=x_i(\theta, 10)$, $i=4,5,6,7$. The coloring represents $x_4=blue$, $x_5=black$, $x_6=red$ and $x_7=green$.}\label{four10}
\end{figure}

% Denote by ${\bf n}(\theta, p)$ the number of solutions of the system~\eqref{rs3.2a},~\eqref{rs3.2b}. Then we have 
% \begin{equation}\label{ns}
% 	{\bf n}(\theta, p)=\left\{
% 	\begin{array}{lllll}
% 		1, \ \ \mbox{if} \ \ \theta\geq 1, \, p>0 \ \ \mbox{or} \ \
% 		\theta\in \big((2\sqrt{2}-1)/7, 1\big), p<\min\{m(\theta), M(\theta)\} \\[2mm]
% 		3, \ \ \mbox{if} \ \ p=m(\theta) \ \ \mbox{or} \ \ p=M(\theta)\\[2mm]
% 		5, \ \ \mbox{if} \ \ (\theta, p)\in \mathcal Q_-\cap \mathcal P\\[2mm]
% 		6, \ \ \mbox{if} \ \ (\theta, p)\in \mathcal Q_0\\[2mm]
% 		7, \ \ \mbox{if} \ \  (\theta, p)\in \mathcal Q_+
% 	\end{array}
% 	\right.
% \end{equation}
% where $(2\sqrt{2}-1)/7$ is unique positive solution of $q(\theta, 0)=0$.

\medskip
Now, for each $x_i(\theta, p)$, using the following condition on the parameters $\theta, p$,
\begin{equation}\label{C2} 
	(1-\theta^{2^p})x_i-\theta^{2^p}(x_i^2+1)\geq 0,
	\end{equation}
we define
	\begin{equation}\label{y4-7}
		y_{i}(\theta, p)=\sqrt{\theta}^{-1}\sqrt{	(1-\theta^{2^p})x_i-\theta^{2^p}(x_i^2+1)}, \ \ i=4,5,6,7.
	\end{equation}
\begin{rk} In \cite{KR-sos}, for $p=1$ it is proven that the Conditions~\eqref{C1} and~\eqref{C2} are satisfied for all values of $\theta$ where the solutions exist. But in case $p\ne 1$ we do not have such a result since the solutions have a very bulky form. Below, we consider two concrete values $p=0.1$ and $p=10$ for any $\theta\in (0,1)$. We note that for these values of $p$, the   Conditions~\eqref{C1} and~\eqref{C2} are satisfied too, see Fig.~\ref{y1-7} and~\ref{y1-710}.
\end{rk}	
\begin{figure}[h]
 \begin{center}
 \includegraphics[width=11cm]{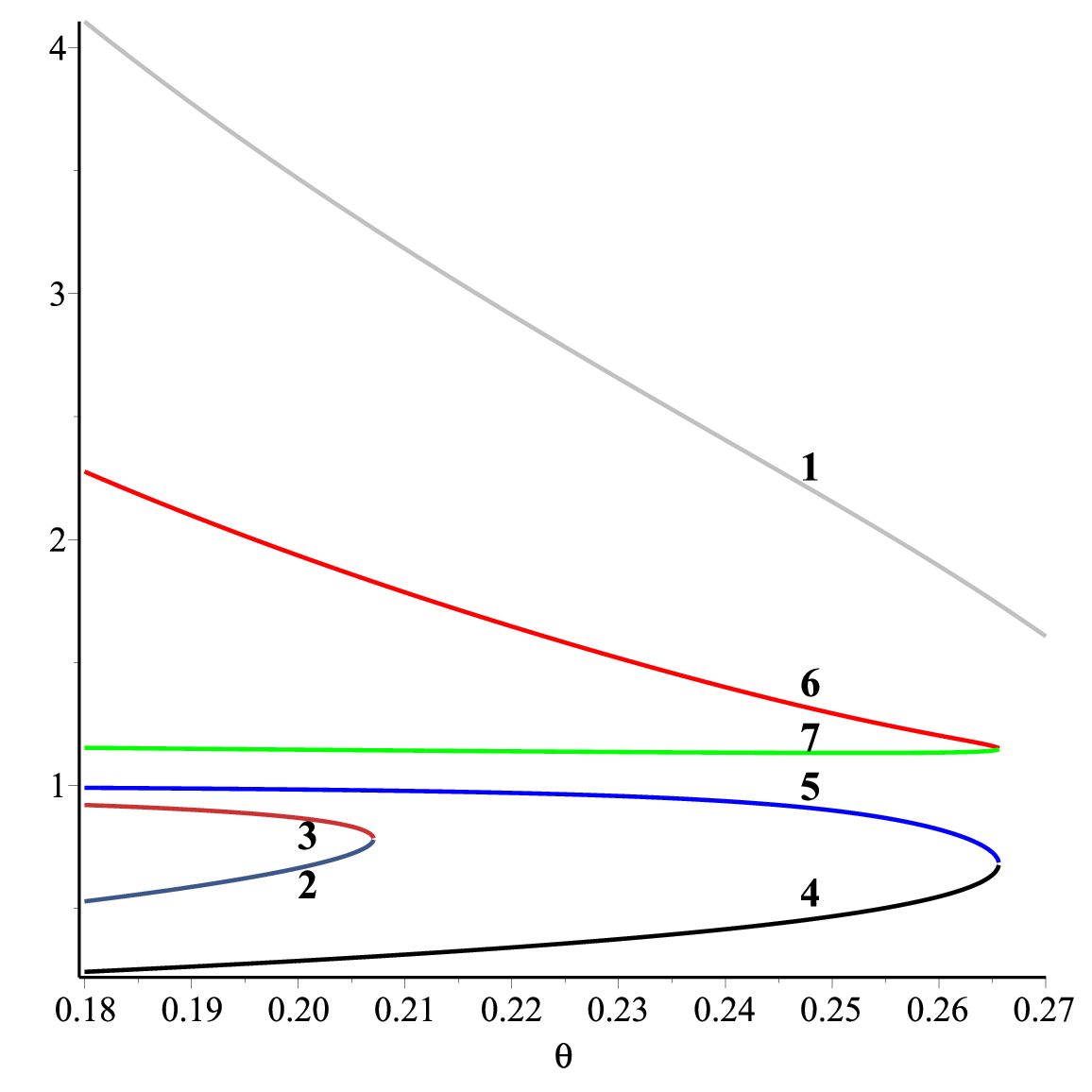}
 \end{center}
 \caption{
 	The graphs of the functions $y_i=y_i(\theta,0.1)$, $i=1,2,...,7$. Here, $y_1=grey$, $y_2=azure$, $y_3=orange$, $y_4=black$, $y_5=blue$, $y_6=red$, and $y_7=green$.}\label{y1-7}
\end{figure}

\begin{figure}[h]
	\begin{center}
\includegraphics[width=11cm]{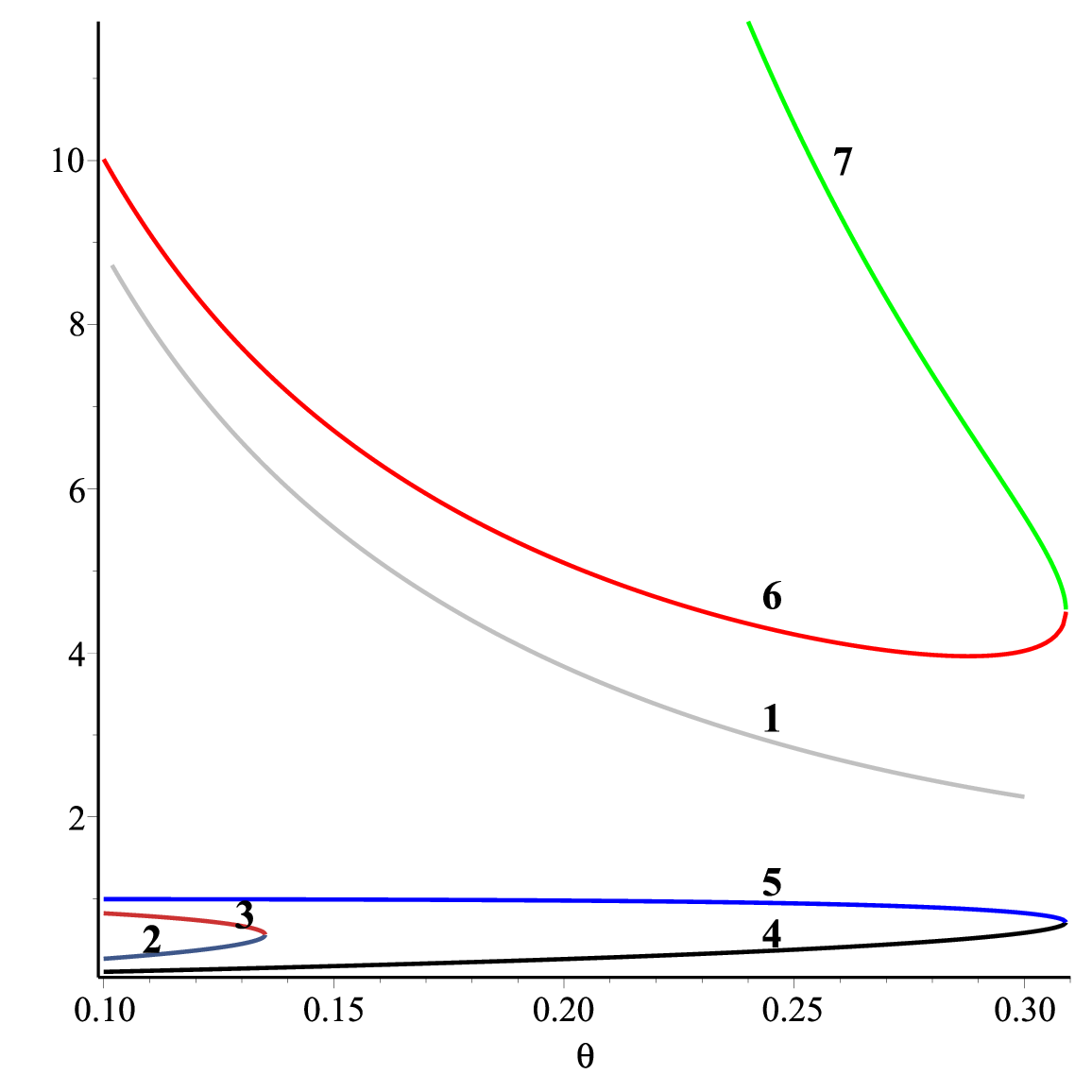}
	\end{center}
	\caption{
		The graphs of the functions $y_i=y_i(\theta,10)$, $i=1,2,...,7$. Here,  $y_1=grey$,$y_2=azure$, $y_3=orange$, $y_4=black$, $y_5=blue$, $y_6=red$, $y_7=green$.}\label{y1-710}
\end{figure}
%In Fig.~\ref{y1-7}, the graphs of all $y_i$, $i=1,2,\dots,7$ are shown for $p=0.1$.
Note that $(2\sqrt{2}-1)/7$ is the unique positive solution of $q(\theta, 0)=0$. Summarizing, we exhibit the full characterization of the solutions as follows.
 \begin{pro}\label{p2} Assume that the Conditions~\eqref{C1} and~\eqref{C2} are satisfied. Then, the 
 set $S(\theta,p)$ of solutions to the system~\eqref{rs3.2a},~\eqref{rs3.2b} 
 changes under variations of the parameters $\theta$ and $p$ as follows:
 \begin{equation}\label{ns}
 	S(\theta, p)=\left\{
 	\begin{array}{lllll}
 		\{v_1=(1, y_1)\}, \ \ \mbox{if} \ \ \theta\geq 1, \, p>0 \ \ \mbox{or} \ \
 		\theta\in \left({2\sqrt{2}-1\over 7}, 1\right), p<\min\{m(\theta), M(\theta)\} \\[2mm]
 	\{v_1, v_4=(x_4, y_4), v_6=(x_6, y_6)\}, \ \ \mbox{if} \ \ p=m(\theta) \ \ \mbox{or} \ \ p=M(\theta)\\[2mm]
 	\{v_1, v_i=(x_i,y_i), i=4,5,6,7\}, \ \ \mbox{if} \ \ (\theta, p)\in \mathcal Q_-\cap \mathcal P\\[2mm]
 	\{v_1, v_i=(x_i, y_i), i=3,4,5,6,7\}, \ \ \mbox{if} \ \ (\theta, p)\in \mathcal Q_0\\[2mm]
 \{v_i=(x_i,y_i), i=1,2,3,4,5,6,7\}, \ \ \mbox{if} \ \  (\theta, p)\in \mathcal Q_+,
 	\end{array}
 	\right.
 \end{equation}
 where $y_i$, $i=1,2,3$ are solutions of~\eqref{y3}, which can be given explicitly by Cardano's formula, $x_1=x_2=x_3=1$, $x_i$ and $y_i$ for $i=4,5,6,7$ are given by~\eqref{x4-7} and~\eqref{y4-7}.
 \end{pro}
We present in Fig.~\ref{four} and~\ref{y1-7} the graphs of the functions mentioned in Proposition~\ref{p2} for the case $p=0.1$ and Fig.~\ref{four10} for the case $p=10$.

\medskip
Denote by $\mu_i$ the TISGM corresponding to $v_i$, $i=1,\dots,7$.
 As an immediate corollary to Propositions \ref{rsp2.1} and \ref{p2} we get the following statement.
\begin{thm}\label{tm} 	
Assume that, for the parameters of the $p$-SOS model, the Conditions~\eqref{C1} and~\eqref{C2} are satisfied. Then, the number of TISGMs $\mathcal M(\theta,p)$
	changes under variations of the parameters $\theta$ and $p$ as follows:
	\begin{equation}\label{nsm}
		\mathcal M(\theta, p)=\left\{
		\begin{array}{lllll}
			1, \ \ \mbox{if} \ \ \theta\geq 1, \, p>0 \ \ \mbox{or} \ \
			\theta\in \big((2\sqrt{2}-1)/7), 1\big), p<\min\{m(\theta), M(\theta)\} \\[2mm]
			3, \ \ \mbox{if} \ \ p=m(\theta) \ \ \mbox{or} \ \ p=M(\theta)\\[2mm]
			5, \ \ \mbox{if} \ \ (\theta, p)\in \mathcal Q_-\cap \mathcal P\\[2mm]
			6, \ \ \mbox{if} \ \ (\theta, p)\in \mathcal Q_0\\[2mm]
			7, \ \ \mbox{if} \ \  (\theta, p)\in \mathcal Q_+
		\end{array}
		\right.
	\end{equation}
\end{thm}

\begin{rk}
Note that the first part of the first line in~\eqref{nsm} refers to the antiferromagnetic $p$-SOS model, which features only one TISGM for all $p>0$. 
\end{rk}

\section{Tree-indexed Markov chains of TISGMs.}\label{sec_5}

First note that a TISGM corresponding to a vector $v=(x,y)\in \mathbb{R}^2$ (which is a solution to the system~\eqref{rs3.2a}, \eqref{rs3.2b}) is a tree-indexed Markov chain with states $\{0,1,2\}$, see \cite[Definition 12.2]{Ge}, and for the transition matrix
\begin{equation}\label{m} {\mathbb P}:=\left(\begin{array}{ccc}
{x^2\over x^2+\theta y^2+\theta^{2^p}}&{\theta y^2\over x^2+\theta y^2+\theta^{2^p}}&{\theta^{2^p}\over x^2+\theta y^2+\theta^{2^p}}\\[3mm]
{\theta x^2\over \theta x^2+y^2+\theta}&{y^2\over \theta x^2+y^2+\theta}&{\theta\over \theta x^2+y^2+\theta}\\[3mm]
{\theta^{2^p}x^2\over \theta^{2^p}x^2+\theta y^2+1}&{\theta y^2\over \theta^{2^p}x^2+\theta y^2+1}&{1\over \theta^{2^p}x^2+\theta y^2+1}
\end{array}
\right).
\end{equation}

For each given solution $(x_i,y_i)$, $i=1,\dots,7$ of the system~\eqref{rs3.2a}, \eqref{rs3.2b}, we need to calculate the eigenvalues of $\mathbb P$. The first eigenvalue is one since we deal with a stochastic matrix, the other two eigenvalues
 \begin{equation}\label{ev}\lambda_j(x_i, y_i, \theta, p), \qquad j=1,2,
\end{equation} 	
 can be found via computer, but they have bulky formulas. 
For example, in the case $x=1$ for $y$ we have up to three possible values, as mentioned in Lemma~\ref{l1}, and the matrix~\eqref{m} has three eigenvalues, 1 and
$$\lambda_1(1, y,\theta,p)={(\theta^{2^p}-2\theta^2+1)y^2\over \theta y^4+(\theta^{2^p}+2\theta^2+1)y^2+2\theta(\theta^{2^p}+1)}\qquad\text{and}\qquad
\lambda_2(1, y,\theta, p)=
{1-\theta^{2^p}\over \theta^{2^p}+\theta y^2+1}.
$$

\subsection{Conditions for non-extremality}
A sufficient condition for non-extremality of a Gibbs measure $\mu$ corresponding
to the matrix ${\mathbb P}$ on a Cayley tree of order $k\geq 1$ is given by the Kesten--Stigum Condition
$k\lambda^2_2>1$, where $\lambda_2$ is the second largest (in absolute value) eigenvalue of ${\mathbb P}$, see~\cite{Ke}.
Using this criterion, in this section, we find the regions of the parameters $\theta$ and $p$ where the
TISGMs $\mu_i$, $i=1, 2, 3, 4, 5, 6, 7$ are not extreme in the set of all Gibbs measures.
Let us denote 
$$\lambda_{\max,i}(\theta,p):=\max\{|\lambda_1(x_i,y_i,\theta,p)|, |\lambda_2(x_i,y_i,\theta,p)|\}, \ \ i=1,\dots,7,$$
$$\eta_i(\theta,p):=2\lambda^2_{\max,i}(\theta,p)-1,\ \ i=1,\dots,7,$$
and
$$\mathbb K_i:=\{(\theta, p)\in (0,1)\times (0, +\infty): \eta_i(\theta,p)>0\}, \ \ i=1, \dots, 7.$$
Then, the Kesten--Stigum Condition provides us with the following criterion.
\begin{pro}\label{tne}
If $(\theta, p)\in \mathbb K_i$ is such that $\mu_i$ exists then, $\mu_i$ is non-extremal.
\end{pro}

Let us illustrate this proposition for the measures $\mu_i$, $i=1,2,3$ for two values of $p$, namely $p=0.1$ and  $p=10$. The precise choice of the values is personal taste. 
\subsubsection{Case: $p=0.1$} We note that $y_1$ exists for any $\theta>0$ (see Lemma~\ref{l1}).
For the case $y_1$, via computer analysis, one can check that there is $\theta_1\approx 0.32$ such that  
$$\lambda_{\max,1}(\theta,0.1)=\left\{\begin{array}{ll}
|\lambda_1(1,y_1,\theta,0.1)|, \ \ \mbox{if} \ \ \theta \in(0, \theta_1)\\[2mm]
|\lambda_2(1,y_1,\theta,0.1)|, \ \ \mbox{if} \ \ \theta\geq \theta_1.
\end{array}\right.$$
To see that the function  $\eta_1(\theta,0.1)$ is monotone increasing for $\theta>\tilde\theta_1\approx 1523.4$ we draw the graph of the function $\eta_1(1/\theta,0.1)$ for $\theta\in (0, 0.001)$ (see Fig.~\ref{ne1i0}). 
\begin{figure}[h]
	\begin{center}
		\includegraphics[width=9cm]{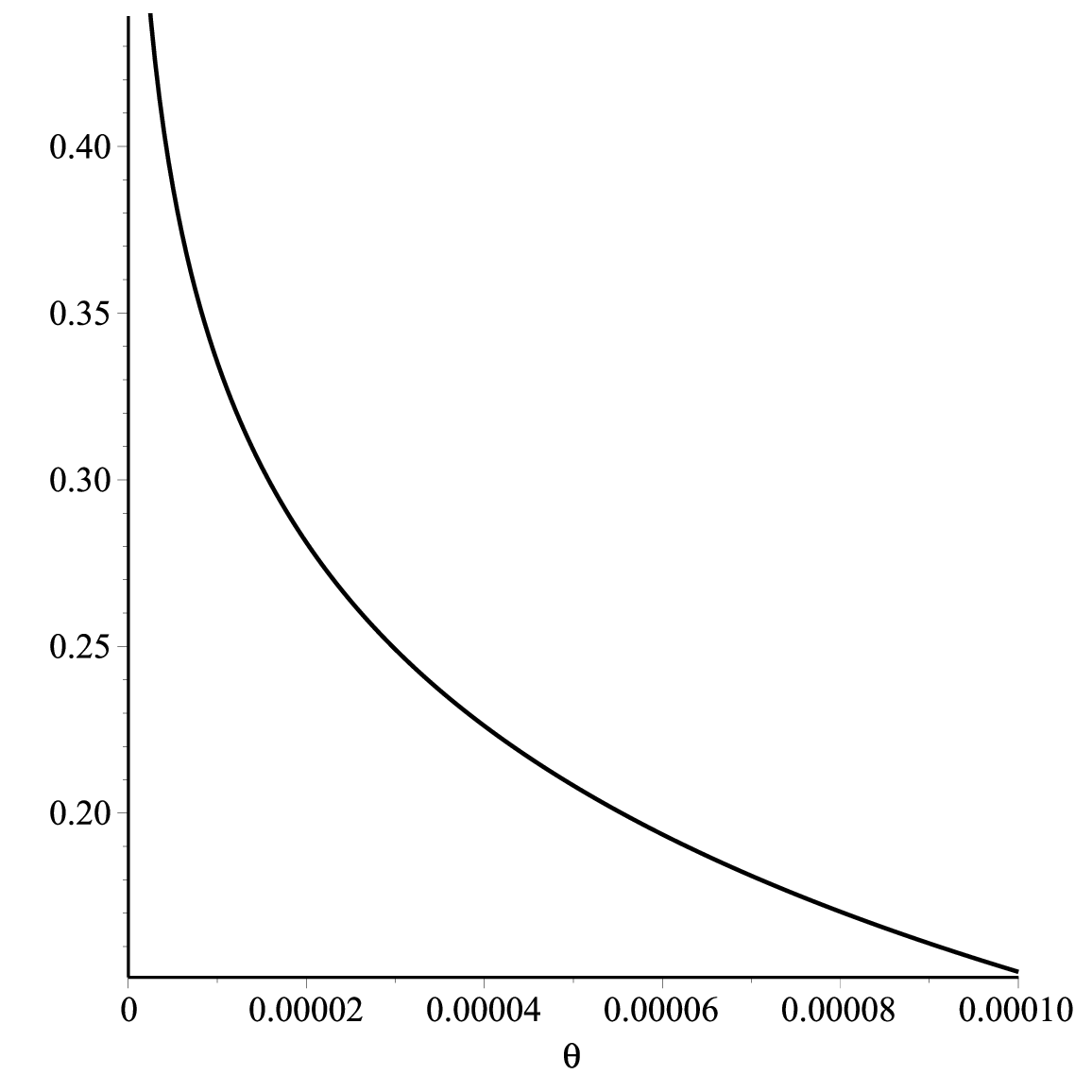}\ \ \
	\end{center}
	\caption{Graph of the function $\eta_1(1/\theta, 0.1)$, for $\theta\in(0, 0.001)$. }\label{ne1i0}
\end{figure}
This leads to the following conclusion. 
\begin{re} Considering Fig.~\ref{ne1} and~\ref{ne1i0} we conclude  that the Kesten--Stigum Condition does not hold for $(1, y_1(\theta,0.1))$, with $\theta\in (0, \tilde \theta_1]$, where $\tilde\theta_1\approx 1523.4$. However, the condition holds for $\theta>\tilde\theta_1$.
\end{re}
\begin{figure}[h]
\begin{center}
\includegraphics[width=7cm]{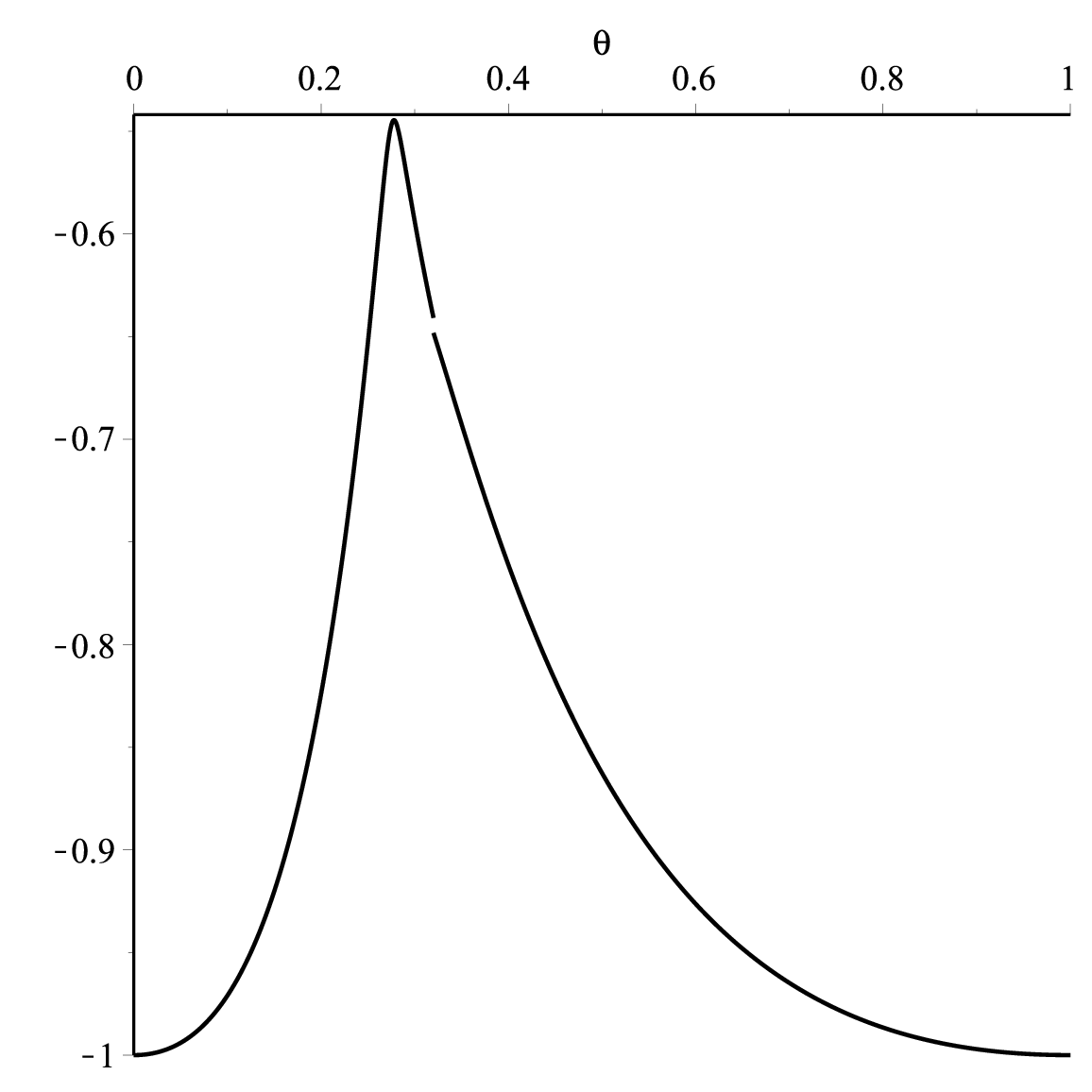} \ \ \ \ \
\includegraphics[width=7cm]{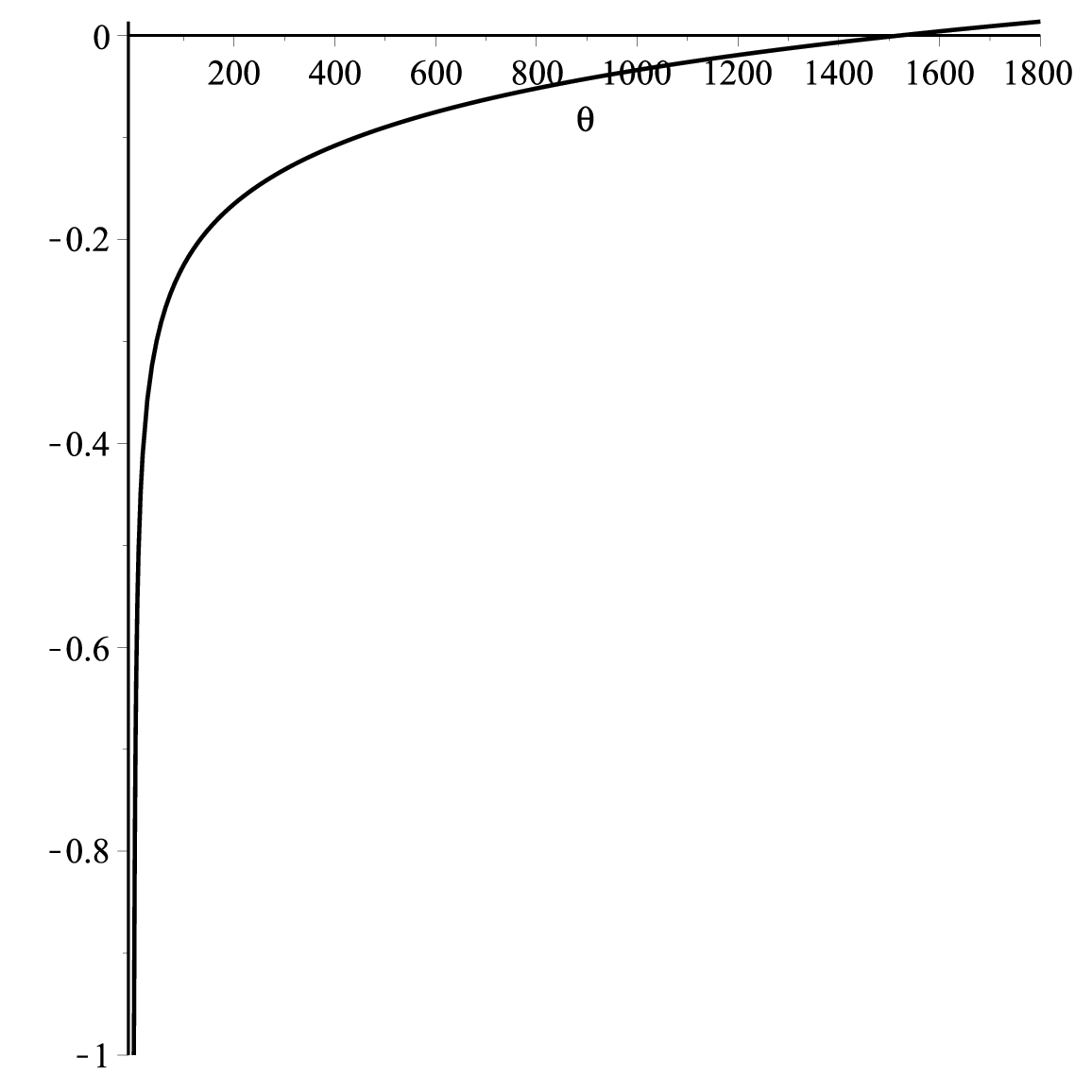}\ \ \
\end{center}
\caption{Graph of the function $\eta_1(\theta,0.1)$, for $\theta\in(0,1)$ (left) and for $\theta>1$ (right). }\label{ne1}
 \end{figure}

Next, from Fig.~\ref{y1-7}, for $p=0.1$ we know that $y_2$ and $y_3$ exist only for $\theta<\theta_2\approx 0.206$ and thus, using again computer algebra, one can see that 
$$\lambda_{\max,i}(\theta,0.1)=|\lambda_2(1,y_i,\theta,0.1)|, \ \ \forall \theta\in (0, \theta_2), i=2,3.$$
 
\begin{re} Considering Fig.~\ref{ne2}, we see that there are $\hat \theta_2, \hat\theta_3<\theta_2$ ($\hat \theta_2\approx 0.175$, $\hat \theta_3\approx 0.139$) such that the Kesten--Stigum Condition holds for solutions  $(1, y_i(\theta,0.1))$, with $\theta\in (0, \hat\theta_i)$, $i=2,3$.
\end{re}
\begin{figure}[h]
	\begin{center}
		\includegraphics[width=7cm]{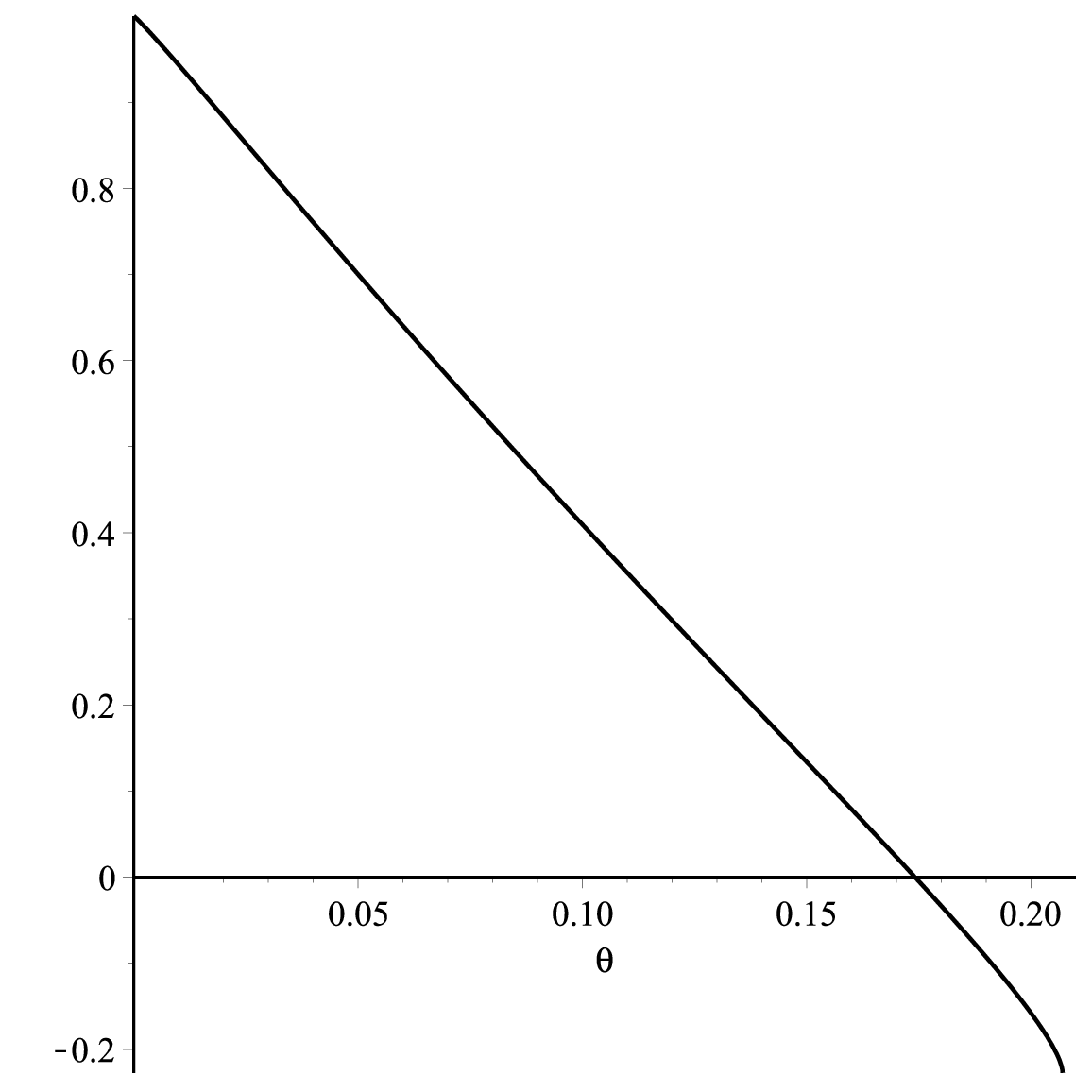} \ \ \ \ \
		\includegraphics[width=7cm]{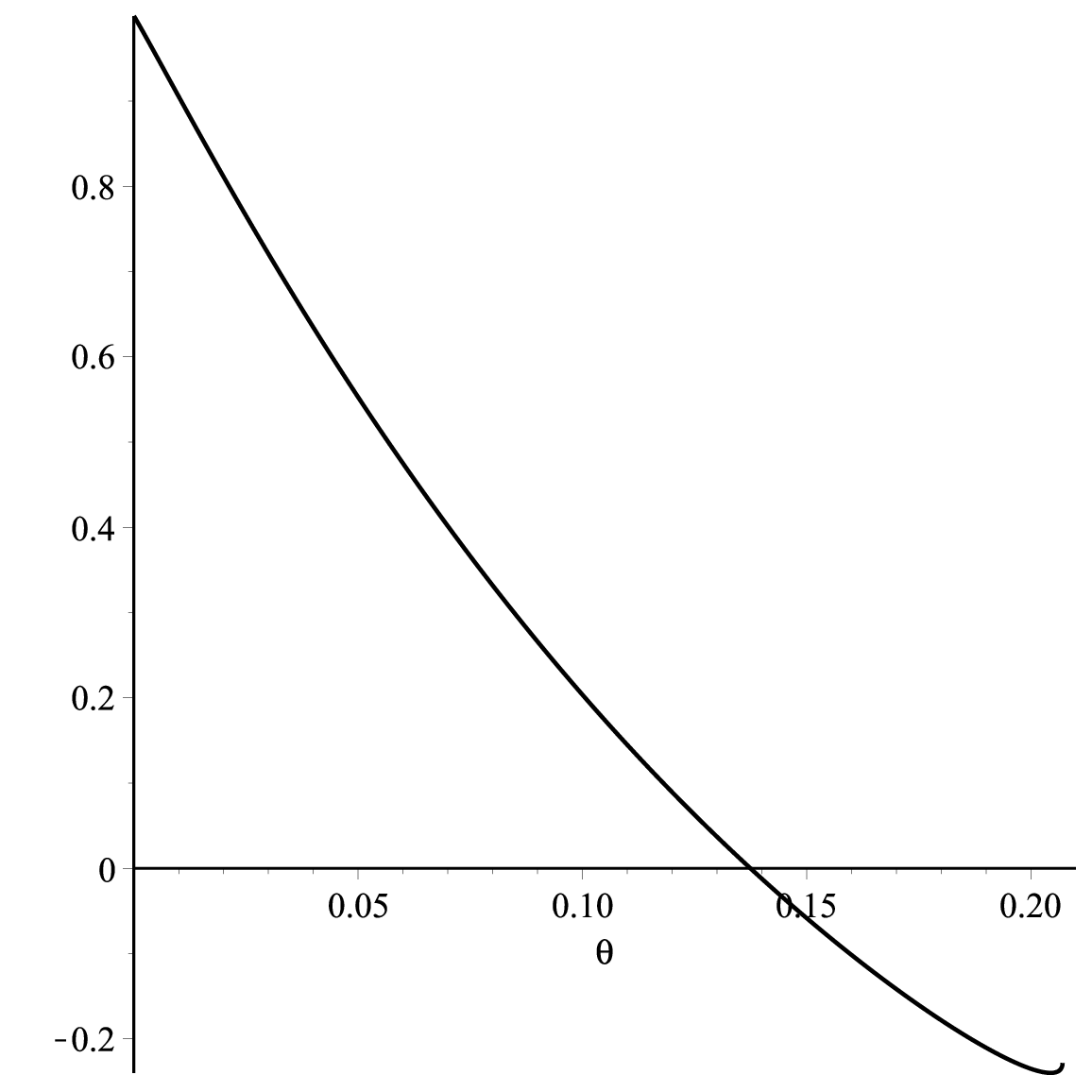}\ \ \
	\end{center}
	\caption{Graphs of the functions $\eta_2(\theta,0.1)$ (left) and  $\eta_3(\theta,0.1)$ (right),  $\theta\in(0, \theta_2)$.}\label{ne2}
\end{figure}

\subsubsection{Case: $p=10$} For $p=10$, in the case $y_1$, by computer analysis, one can check that   
$$\lambda_{\max,1}(\theta,10)=|\lambda_2(1,y_1,\theta,10)|.$$
\begin{re} Fig.~\ref{ne10} shows that the Kesten--Stigum Condition never holds for the solution $(1, y_1(\theta,10))$, with $\theta\in (0,1]$. But it holds for $\theta>1$.
\end{re}
\begin{figure}[h]
	\begin{center}
		\includegraphics[width=9cm]{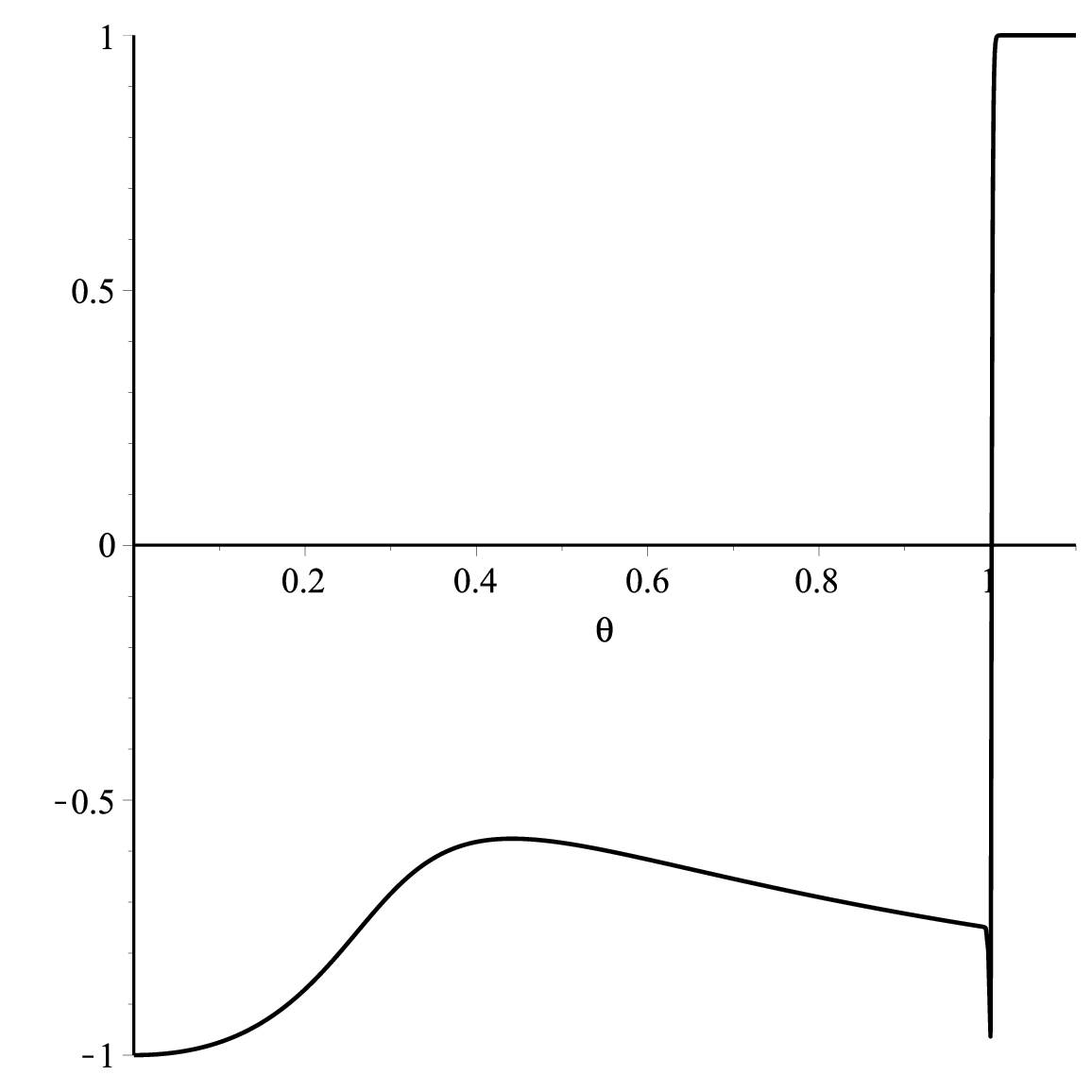} \ \ \ \ \
	\end{center}
	\caption{Graph of the function $\eta_1(\theta,10)$, $\theta\in(0,1)$. }\label{ne10}
\end{figure}

The fact that the threshold for $p=10$ is precisely given by $\theta=1$ is remarkable. We also note that in the antiferromagnetic case $\theta>1$, there is a unique TISGM which is however not extremal in case $p=10$. 
From Fig.~\ref{y1-7} we know that $y_2$ and $y_3$ exist only for $\theta<\theta'_2\approx 0.136$ and hence, using computer algebra, one can see that 
$$\lambda_{\max,i}(\theta,10)=|\lambda_2(1,y_i,\theta,10)|, \ \ \forall \theta\in (0, \theta'_2), i=2,3.$$
\begin{re}  Considering Fig.~\ref{ne20}, we see that the Kesten--Stigum Condition always holds for the solution $(1, y_i(\theta,10))$, with $\theta\in (0, \theta'_2)$, $i=2,3$.
\end{re}
\begin{figure}[h]
	\begin{center}
		\includegraphics[width=7cm]{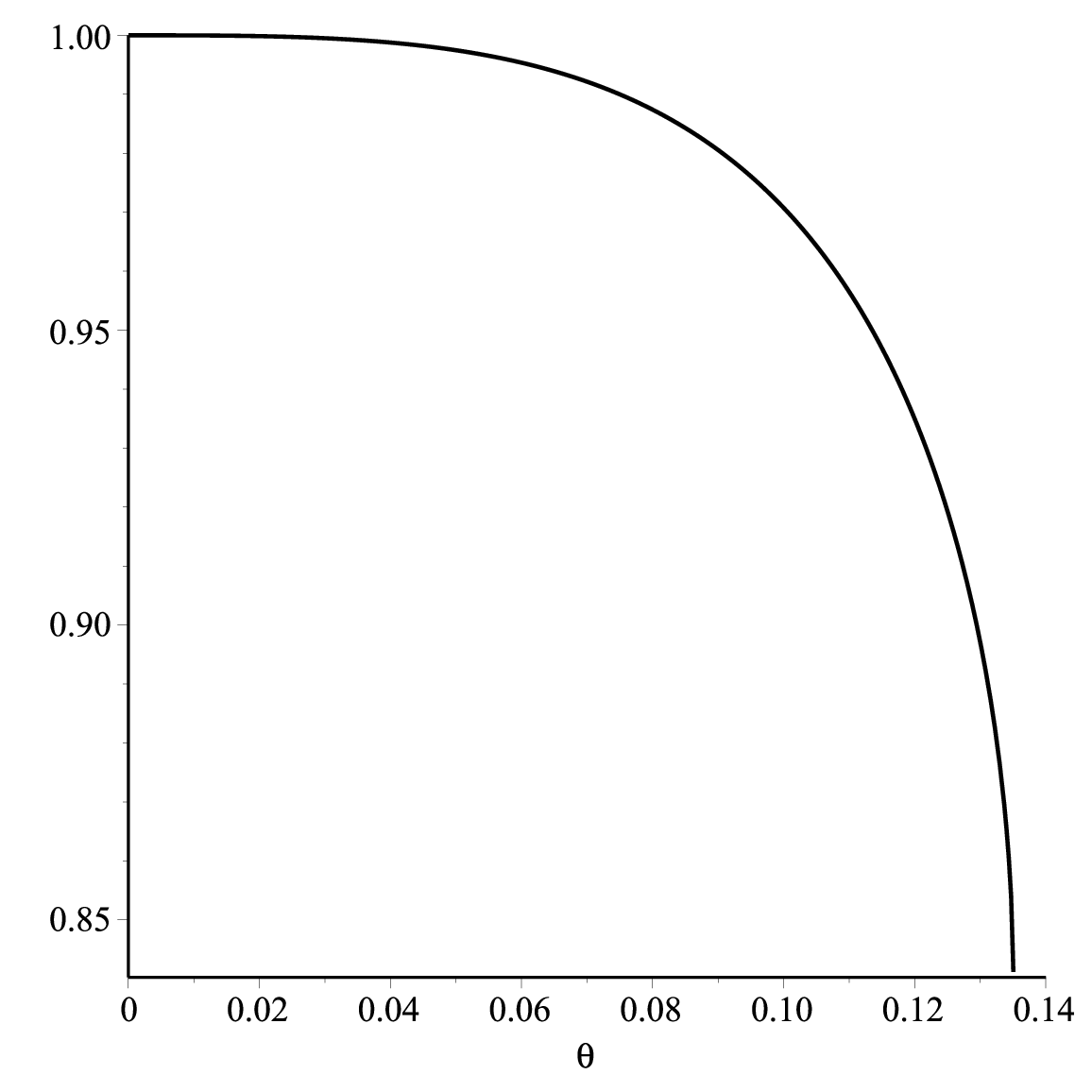} \ \ \ \ \
		\includegraphics[width=7cm]{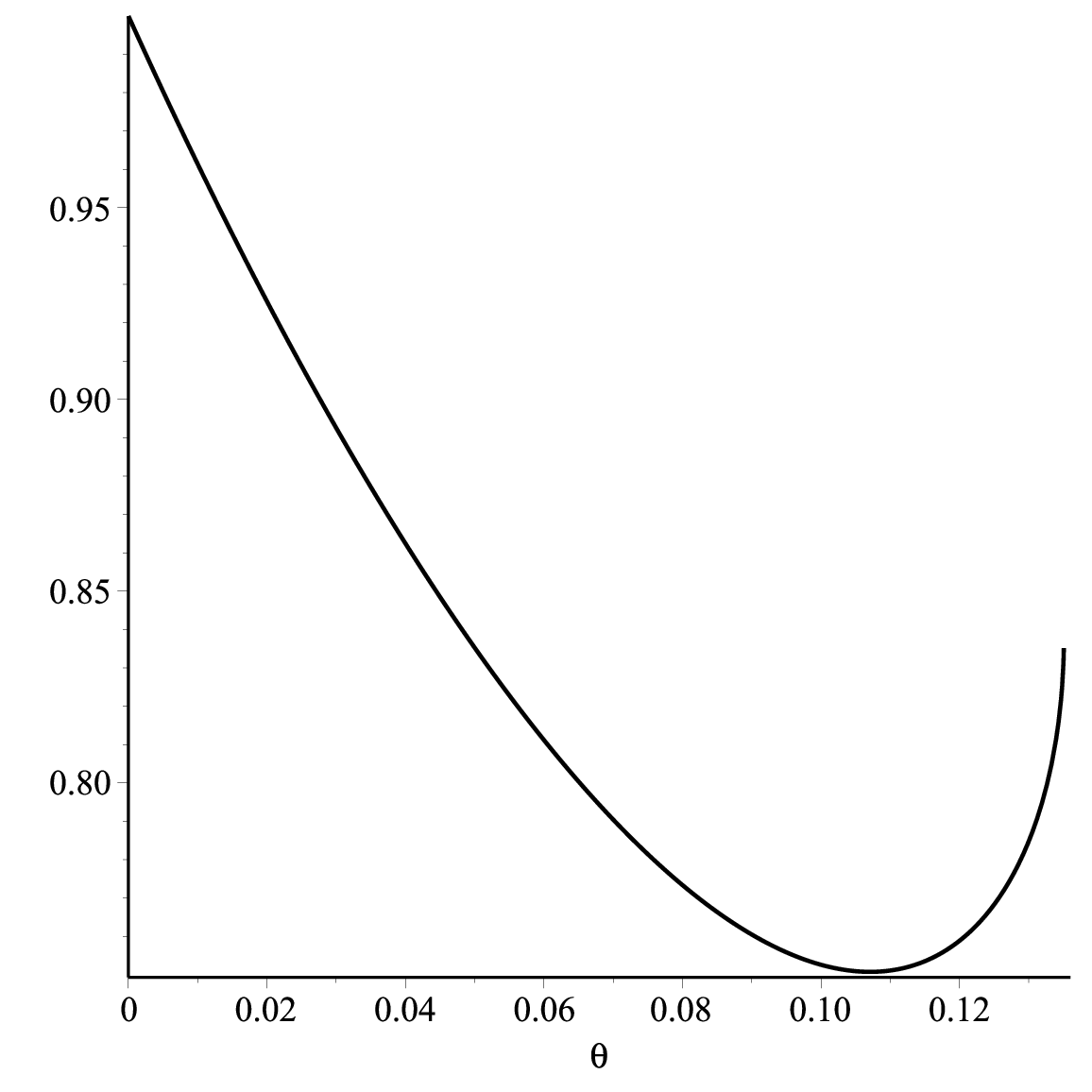}\ \ \
	\end{center}
	\caption{The graphs of the functions $\eta_2(\theta,10)$ (left) and  $\eta_3(\theta,10)$ (right),  $\theta\in(0, \theta'_2)$.}\label{ne20}
\end{figure}

 \subsection{Conditions for extremality}
 In this section we find sufficient conditions for extremality (or non-reconstructability in information-theoretic language \cite{MSW,Mos2,Mos,Sly11}) of TISGMs for the 3-state $p$-SOS model,
 depending on parameters $\theta$, $p$ and the boundary law.  We shall consider the TISGMs: $\mu_i$, $i=1,\dots, 7$.
 In order to check extremality, we will use a result of~\cite{MSW} to establish a bound for reconstruction impossibility that corresponds to the matrix~\eqref{m}) of a solution $(x_i, y_i)$, $i=1, \dots, 7$.

 Let us start by recalling some definitions from~\cite{MSW}. Considering finite complete subtrees $\mathcal T$ that are {\em initial} with respect to the Cayley tree $\Gamma^k$, i.e., share the same root. If $\mathcal T$ has depth $d$, i.e., the vertices of $\mathcal T$ are within distance $\leq d$ from the root, then it has $(k^{d+1}-1)/(k-1)$ vertices, and its boundary $\partial \mathcal T$ consists of the neighbors (in $\Gamma^k\setminus \mathcal T$) of its vertices, i.e., $|\partial \mathcal T|=k^{d+1}$.  We identify subgraphs of $\mathcal T$ with their vertex sets and write $E(A)$ for the edges within a subset $A$ and $\partial A$ for the boundary of $A$, i.e., the neighbors of $A$ in $(\mathcal T\cup \partial\mathcal T)\setminus A)$.

 Consider Gibbs measures $\{\mu^\tau_{{\mathcal T}}\}$, where the boundary condition $\tau$
 is fixed and $\mathcal T$ ranges over all initial finite complete subtrees of $\Gamma^k$.
  For a given subtree $\mathcal T$ of $\Gamma^k$ and a vertex $v\in\mathcal T$, we write $\mathcal T_v$ for
 the (maximal) subtree of $\mathcal T$ rooted at $v$. When $v$ is not the root of $\mathcal T$, let $\mu_{\mathcal T_v}^s$
 denote the (finite-volume) Gibbs measure in which the
 parent of $v$ has its spin fixed to $s$ and the configuration on the bottom boundary of ${\mathcal T}_v$
 (i.e., on $\partial {\mathcal T}_v\setminus \{\mbox{parent\ \ of}\ \ v\}$) is specified by $\tau$.

  For two measures $\mu_1$ and $\mu_2$ on $\Omega$, $\|\mu_1-\mu_2\|_v$ denotes the variational distance between the projections of $\mu_1$ and $\mu_2$ onto the spin at $v$, i.e.,
 $$\|\mu_1-\mu_2\|_v:={1\over 2}\sum_{i=0}^2|\mu_1(\sigma(v)=i)-\mu_2(\sigma(v)=i)|.$$
 Let $\eta^{v,s}$ be the
 configuration $\eta$ with the spin at $v$ set to $s$.
 Following~\cite{MSW}, we define
 $$\kappa:= \kappa(\mu)=\sup_{v\in\Gamma^k}\max_{v,s,s'}\|\mu^s_{{\mathcal T}_v}-\mu^{s'}_{{\mathcal T}_v}\|_v\quad\text{ and }\quad\gamma:=\gamma(\mu)=\sup_{A\subset \Gamma^k}\max\|\mu^{\eta^{w,s}}_A-\mu^{\eta^{w,s'}}_A\|_v,$$
 where the maximum on the right-hand side is taken over all boundary conditions $\eta$, all sites $w\in \partial A$, all neighbors $v\in A$ of $w$, and all spins $s, s'\in \{0,1,2\}$. We apply~\cite[Theorem 9.3]{MSW} which goes as follows.
 \begin{thm} For an arbitrary (ergodic\footnote{Ergodicity here means irreduciblity and aperiodicity. In this case, we have a unique stationary distribution $\pi=(\pi_1,\dots,\pi_q)$ with $\pi_i>0$ for all $i$.} and permissive\footnote{Permissive here means that, for arbitrary finite $A$ and boundary condition outside $A$ being $\eta$, the conditioned Gibbs measure on $A$, corresponding to the channel, is positive for at least one configuration.}) channel ${\mathbb P}=(P_{ij})_{i,j=1}^q$
 on a tree, the reconstruction of the corresponding tree-indexed Markov chain
 is impossible if $k\kappa\gamma<1$.
 \end{thm}

  Since each TISGM $\mu$ corresponds to a solution $(x,y)$ of the system of equations~\eqref{rs3.2a} and~\eqref{rs3.2b}, we can write $\gamma(\mu)=\gamma(x,y)$ and $\kappa(\mu)=\kappa(x,y)$.
  
 It is easy to see that the channel ${\mathbb P}$ corresponding to a TISGM of the $p$-SOS model is ergodic and permissive. Thus the criterion of {\it extremality} of a TISGM is $k\kappa\gamma<1$.
 Note that $\kappa$ has the particularly simple form (see~\cite{MSW})
 \begin{equation}\label{ka}
 \kappa={1\over 2}\max_{i,j}\sum_l|P_{il}-P_{jl}|
 \end{equation}
 and $\gamma$ is a constant that does not have a clean general formula, but it can be estimated.

 \subsubsection{Estimation of $\gamma$.}
To estimate the constant $\gamma(x_i,y_i)$
 depending on the boundary law labeled by $i$, for our model, we prove several lemmas.
  \begin{lemma}\label{pd} Recall the matrix ${\mathbb P}$ given by~\eqref{m} and denote by
 $\mu=\mu(\theta,p)$ the corresponding Gibbs measure. Then, for any subset $A\subset {\mathcal T}$, (where $\mathcal T$ is a complete subtree of $\Gamma^k$)
 any boundary configuration $\eta$, any pair of spins
 $(s_1, s_2)$, any site $w\in \partial A$, and any neighbor $v\in A$ of $w$, we have
 $$\|\mu^{\eta^{w,s_1}}_A-\mu^{\eta^{w,s_2}}_A\|_v\leq
\max\{p^0(0)-p^1(0),p^0(0)-p^2(0), |p^i(1)-p^j(1)|, p^2(2)-p^0(2), p^2(2)-p^1(2)\},$$
 where $p^{t}(s):=\mu^{\eta^{w,t}}_A(\sigma(v)=s)$.
\end{lemma}
\begin{proof} Denote $p_s=\mu^{\eta^{w,\text{free}}}_A(\sigma(v)=s)$, $s=0,1,2$.   By definition of the matrix ${\mathbb P}$, we have
$$
 p^0(0)={x^2p_0\over x^2p_0+\theta y^2p_1+\theta^{2^p} p_2}, \ \
  p^0(1)={\theta y^2p_1\over x^2p_0+\theta y^2p_1+\theta^{2^p} p_2},\ \
   p^0(2)={\theta^{2^p}p_2\over x^2p_0+\theta y^2p_1+\theta^{2^p} p_2};$$
    \begin{equation}\label{ppp}
  p^1(0)={\theta x^2p_0\over \theta x^2p_0+ y^2p_1+\theta p_2}, \ \
     p^1(1)={y^2p_1\over \theta x^2p_0+ y^2p_1+\theta p_2},\ \
      p^1(2)={\theta p_2\over \theta x^2p_0+ y^2p_1+\theta p_2};
    \end{equation}
    $$
       p^2(0)={\theta^{2^p}x^2p_0\over \theta^{2^p} x^2p_0+\theta y^2p_1+ p_2}, \ \
        p^2(1)={\theta y^2p_1\over \theta^{2^p} x^2p_0+\theta y^2p_1+ p_2},\ \
         p^2(2)={p_2\over \theta^{2^p} x^2p_0+\theta y^2p_1+p_2},
        $$
and hence, the  proposition follows from the following Lemma~\ref{lem} and Lemma~\ref{lK}.\end{proof}

 \begin{lemma}\label{lem} If $\theta<1$ then
 \begin{itemize}
 \item[a)]  $p^{0}(0)\geq \max\{p^{1}(0), p^2(0)\}$;

 \item[b)] If $p\leq 1$, then $p^{1}(1)\geq \max\{p^{0}(1),  p^2(1)\}$. If $p>1$, then, the values of $p^0(1)$, $p^{1}(1)$ and $p^2(1)$ may have any order depending on $(p_0,p_2)$.
 \item[c)]  $p^{2}(2)\geq \max\{p^{0}(2), p^1(2)\}$.
 \end{itemize}
 \end{lemma}
 \begin{proof}  We shall prove some of the inequalities (all others are proved similarly):
 
 a) By the Formula~\eqref{ppp}, we get
\begin{align*}
p^0(0)-p^1(0)&=x^2p_0{(1-\theta^2)y^2p_1+(1-\theta^{2^p})\theta p_2\over (x^2p_0+\theta y^2p_1+\theta^{2^p} p_2)(\theta x^2p_0+ y^2p_1+\theta p_2)},\\
p^{0}(0)-p^{2}(0)&=x^2p_0{\theta(1-\theta^{2^p})y^2p_1+(1-\theta^{2^{p+1}} p_2)\over \left(x^2p_0+\theta y^2p_1+\theta^{2^p} p_2\right)\left(\theta^{2^p} x^2p_0+ \theta y^2p_1+p_2\right)},
\end{align*}
and both are positive iff $\theta<1$.

b) Consider 
\begin{align*}
p^{1}(1)-p^{0}(1)&=y^2p_1{(1-\theta^2)p_0x^2 +(\theta^{2^p}-\theta^2)p_2\over \left(\theta x^2p_0+ y^2p_1+\theta p_2\right)\left(x^2p_0+\theta y^2p_1+\theta^{2^p}p_2\right)},\\
p^{1}(1)-p^{2}(1)&=y^2p_1{(\theta^{2^p}-\theta^2)p_0x^2 +(1-\theta^2)p_2\over \left(\theta x^2p_0+ y^2p_1+\theta p_2\right)\left(\theta^{2^p}x^2p_0+\theta y^2p_1+p_2\right)},
\end{align*}
   which are non-negative if $\theta<1$ and $p\leq 1$. Moreover, if $p>1$, then, for $\theta<1$, we have $\theta^{2^p}-\theta^2<0$, then assuming $p^{1}(1)-p^{2}(1)<0$, that is $(\theta^{2^p}-\theta^2)p_0x^2 +(1-\theta^2)p_2<0$, we get the condition $p_0>(1-\theta^2)p_2/((\theta^2-\theta^{2^p})x^2)$. It is easy to see that there are values $p>1$ and $\theta<1$ such that all possible inequalities may hold. 
   
   c) Similar to the case a).
   \end{proof}

Next, if $\theta<1$ then we have that
 $$\max_{i,j,k}\left\{|p^i(k)-p^j(k)|\right\}=\max_{i,j}\{p^0(0)-p^1(0),p^0(0)-p^2(0), |p^i(1)-p^j(1)|, p^2(2)-p^0(2), p^2(2)-p^1(2)\}.$$
 Indeed, the case $k=0$ and $k=2$ follow from Lemma~\ref{lem}. For $k=1$ some differences can be reduced to the case $k=0$ or $k=2$, by the following equality
    $$ p^i(1)-p^j(1)= p^j(0)-p^i(0)+p^j(2)-p^i(2).$$
  
Let us give an upper bound of $|p^i(k)-p^j(k)|$, for the maximal ones mentioned above. For $(p_0,p_1,p_2)$ (i.e., a probability distribution on $\{0,1,2\}$) denote $t=p_0$, $u=p_2$, $0\leq t+u\leq 1$ and define the following functions
\begin{align*}
f(t,u,\theta,p)&=p^0(0)-p^2(0)\\
&={x^2t\over (x^2-\theta y^2)t+\theta(\theta^{2^p-1}-y^2)u+\theta y^2}-{x^2\theta^{2^p}t\over \theta(\theta^{2^p-1} x^2-y^2)t+(1-\theta y^2)u+\theta y^2}\\
\varphi(t,u,\theta,p)&= p^0(0)-p^1(0)\\
&={x^2t\over (x^2-\theta y^2)t+(\theta^{2^p}-\theta y^2)u+\theta y^2}-
 {\theta x^2t\over (\theta x^2-y^2)t+(\theta-y^2)u+y^2},\\
 \psi(t,u,\theta,p)&= p^1(1)-p^0(1)\\
 &={y^2u\over \theta(x^2-1)t+(y^2-\theta)u+\theta}-
 {\theta y^2u\over (x^2-\theta^{2^p})t+(\theta y^2-\theta^{2^p})u+\theta^{2^p}},\\
g(t,u,\theta,p)&=p^2(2)-p^0(2)\\
&={u\over \theta(\theta^{2^p-1} x^2-y^2)t+(1-\theta y^2)u+\theta y^2}-{\theta^{2^p} u\over (x^2-\theta y^2)t+\theta(\theta^{2^p-1}- y^2)u+\theta y^2}.
\end{align*}

 \begin{lemma}\label{lK} If $\theta< 1$ then 
 $$\max\{|f(t,u,\theta,p)|, |\varphi(t,u,\theta,p)|,  |\psi(t,u,\theta,p)|, |g(t,u,\theta,p)|\}\leq {1-\theta^{2^p}\over 1+\theta^{2^p}}.$$
 \end{lemma}
 \begin{proof} We present our calculation only for the function $f$, the other functions are checked similarly. To find the maximal value of the function $f$ we have to solve the following system
 \begin{equation}\label{hu}
 \begin{split}
  f_u'(t,u,\theta,p)&={\theta x^2t(y^2-\theta^{2^p-1})\over ((x^2-\theta y^2)t+\theta(\theta^{2^p-1}-y^2)u+\theta y^2)^2}\\
&\qquad +{\theta^{2^p} x^2t(1-\theta y^2)\over (\theta(\theta^{2^p-1} x^2-y^2)t+(1-\theta y^2)u+\theta y^2)^2}=0,
 \end{split}
 \end{equation}
  \begin{equation}\label{ht}
   \begin{split}
   f_t'(t,u,\theta,p)&={\theta x^2u(\theta^{2^p-1}-y^2)+\theta x^2y^2\over ((x^2-\theta y^2)t+\theta(\theta^{2^p-1}-y^2)u+\theta y^2)^2}\\
  &\qquad-{\theta^{2^p} x^2u(1-\theta y^2)+\theta^{2^p+1}x^2y^2\over (\theta(\theta^{2^p-1} x^2-y^2)t+(1-\theta y^2)u+\theta y^2)^2}=0.
   \end{split}
  \end{equation}
  From~\eqref{hu} one has either $t=0$, or if $t\ne 0$
  we note that if $y^2=1/\theta$ then $y^2=\theta^{2^p-1}$, i.e., $\theta=1$.
 So we can assume $y^2\ne 1/\theta$.  Then from~\eqref{hu}
  (for $t\ne 0$) and~\eqref{ht} we get
 $${\theta^{2^p-1}-y^2\over 1-\theta y^2}= {(\theta^{2^p-1}-y^2)u+y^2\over (1-\theta y^2)u+\theta y^2},$$
 which is possible only iff $\theta=1$. So it remains to check only the case $t=0$, which gives a minimum ($=0$) of the function $f$. Hence the maximal value of $f$ is reached on the boundary of the set
 $\{(t,u)\in [0,1]^2: t+u\leq 1\}$.  
 We note that similar results hold for the function $\varphi$ too.  
 We discuss the three line segments
 of the boundary separately:

 {\it Case: $t=0$.} In this case it was already
 mentioned above that the function has a minimum which is equal to zero.

{\it Case: $u=0$.} In this case simple calculations show that
\begin{align*}
\max  f(t,0,\theta,p) &=f\left({y^2\over \theta^{2^{p-1}-1}x^2+y^2},0,\theta,p\right)={1-\theta^{2^{p-1}}\over 1+\theta^{2^{p-1}}}\text{ and }\\
\max  \varphi(t,0,\theta,p) &=\varphi\left({y^2\over x^2+y^2},0,\theta,p\right)={1-\theta\over 1+\theta}.
\end{align*}

{\it Case: $t+u=1$.} In this case we have 
\begin{align*}
\max  f(t,1-t,\theta,p) &=f\left({1\over 1+x^2},{x^2\over 1+x^2},\theta,p\right)={1-\theta^{2^p}\over 1+\theta^{2^p}}\text{ and}\\
\max  \varphi(t,1-t,\theta,p) &=\varphi\left({\theta^{2^{p-1}}\over \theta^{2^{p-1}}+x^2},{x^2\over \theta^{2^{p-1}}+x^2},\theta,p\right)={1-\theta^{2^{p-1}}\over 1+\theta^{2^{p-1}}}.
\end{align*}
Similarly for $\psi$ one can show that 
$$|\psi(t,u,\theta,p)|\leq \max\left\{{|1-\theta|\over 1+\theta},  {\left|1-\theta^{2^{p-1}-1}\right|\over 1+\theta^{2^{p-1}-1}}\right\}.$$
Next, for $\theta<1$ and $t>-1$ consider the following function
$$\Theta(t)={1-\theta^t\over 1+\theta^t}.$$
It is easy to check that this function is monotone increasing and therefore we have 
$$ \max\left\{{1-\theta\over 1+\theta},  {\left|1-\theta^{2^{p-1}-1}\right|\over 1+\theta^{2^{p-1}-1}}, {1-\theta^{2^{p-1}}\over 1+\theta^{2^{p-1}}}\right\}\leq  
	{1-\theta^{2^p}\over 1+\theta^{2^p}}.$$ This completes the proof for $f$. 
 For $g$ the proof is very similar.
 \end{proof}

In the following proposition we now present our bound on $\gamma$.
 \begin{pro} \label{peg} Independent of the possible
 values of $(x,y)$ (i.e., the solutions to the system~\eqref{rs3.2a} and \eqref{rs3.2b}) for  $\theta<1$, $p>0$ we have
 \begin{equation}
 \label{gam}
 \gamma(x,y)\leq  {1-\theta^{2^p}\over 1+\theta^{2^p}}.
 \end{equation}
 \end{pro}
 \begin{proof}
 This is a corollary of the above-mentioned lemmas.
  \end{proof}

\subsubsection{Computation of $\kappa$.} Now we shall compute the constant $\kappa$.
Since $(x,y)$ is a solution to the system~\eqref{rs3.2a}, \eqref{rs3.2b}, the matrix~\eqref{m}  can be written in the following form
\begin{equation}\label{m1} {\mathbb P}={1\over Z}\left(\begin{array}{ccc}
x&\theta y^2/x&\theta^{2^p}/x\\[2mm]
\theta x^2/y&y&\theta/y\\[2mm]
\theta^{2^p}x^2&\theta y^2&1
\end{array}
\right),
\end{equation}
where $Z=\theta^{2^p}x^2+\theta y^2+1.$
  Using~\eqref{ka} and~\eqref{m1}, we get
  \begin{equation}\label{kac}
   \kappa(x,y) =
  {1\over 2}\max_{i,j}\sum_{l=0}^2|P_{il}-P_{jl}|=
  {1\over 2Z}\max\left\{{x^2|y-\theta x|+y^2|x-\theta y|+|\theta^{2^p}y-\theta x|\over xy},\right.\end{equation}
  $$\left.
 {x^2|1-\theta^{2^p} x|+\theta y^2|1-x|+|\theta^{2^p}-x|\over x},\
  {x^2|\theta-\theta^{2^p}y|+y^2|1-\theta y|+|\theta- y|\over y}\right\}, 
 $$
 where $Z=\theta^{2^p}x^2+\theta y^2+1$.
 We are interested in computing $\kappa(x,y)$ for 
 $$(x,y)\in \{(1,y_1),(1,y_2),(1,y_3),(x_4,y_4),(x_5,y_5),
 (x_6,y_6),(x_7,y_7)\}.$$
Since we have an explicit formula for the solutions $(x,y)$ (mentioned in Section~\ref{sec_3}), the value of $\kappa(x,y)$ will be a function of the parameters $\theta, p$. 
Unfortunately, the explicit formulas for the solutions are very bulky, so we start with $x=1$.

{\bf Case $x=1$.} In this case we have 
$$
\kappa(1,y)= {1\over 2Z_1}\max\left\{{|y-\theta|+y^2|1-\theta y|+|\theta^{2^p}y-\theta|\over y}, 2|1-\theta^{2^p}|\right\},$$
where $Z_1=\theta^{2^p}+\theta y^2+1$ and $y$ is a solution to~\eqref{y3}.

\medskip
{\bf Subcase: $p=0.1$ and for solution $y_1$:}  For the solution $y_1$ of~\eqref{y3} from~\eqref{kac}, by computer analysis, one can see that there exists $\hat \theta_1\approx 0.335$ such that 
\begin{equation}\label{k1}
\kappa(1,y_1)={1\over 2Z_1}\left\{\begin{array}{ll}
	{|y_1-\theta|+y_1^2|1-\theta y_1|+|\theta^{2^p}y_1-\theta|\over y_1}, \ \ \mbox{if} \ \ \theta\in (0, \hat \theta_1)\\[2mm]
	 2|1-\theta^{2^p}|, \ \ \mbox{if} \ \ \theta\geq \hat \theta_1.
 \end{array}\right.
\end{equation}
Denote
$$U_1(\theta, p)=2{1-\theta^{2^p}\over 1+\theta^{2^p}}\kappa(1,y_1)-1.$$
\begin{re} Fig.~\ref{e1} shows that the extremality condition holds for the solution $(1, y_1(\theta,0.1))$, with $\theta\in (0, \theta^*_1)$, where $\theta_1^*\approx 19.08$.
\end{re}
\begin{figure}[h]
	\begin{center}
		\includegraphics[width=7cm]{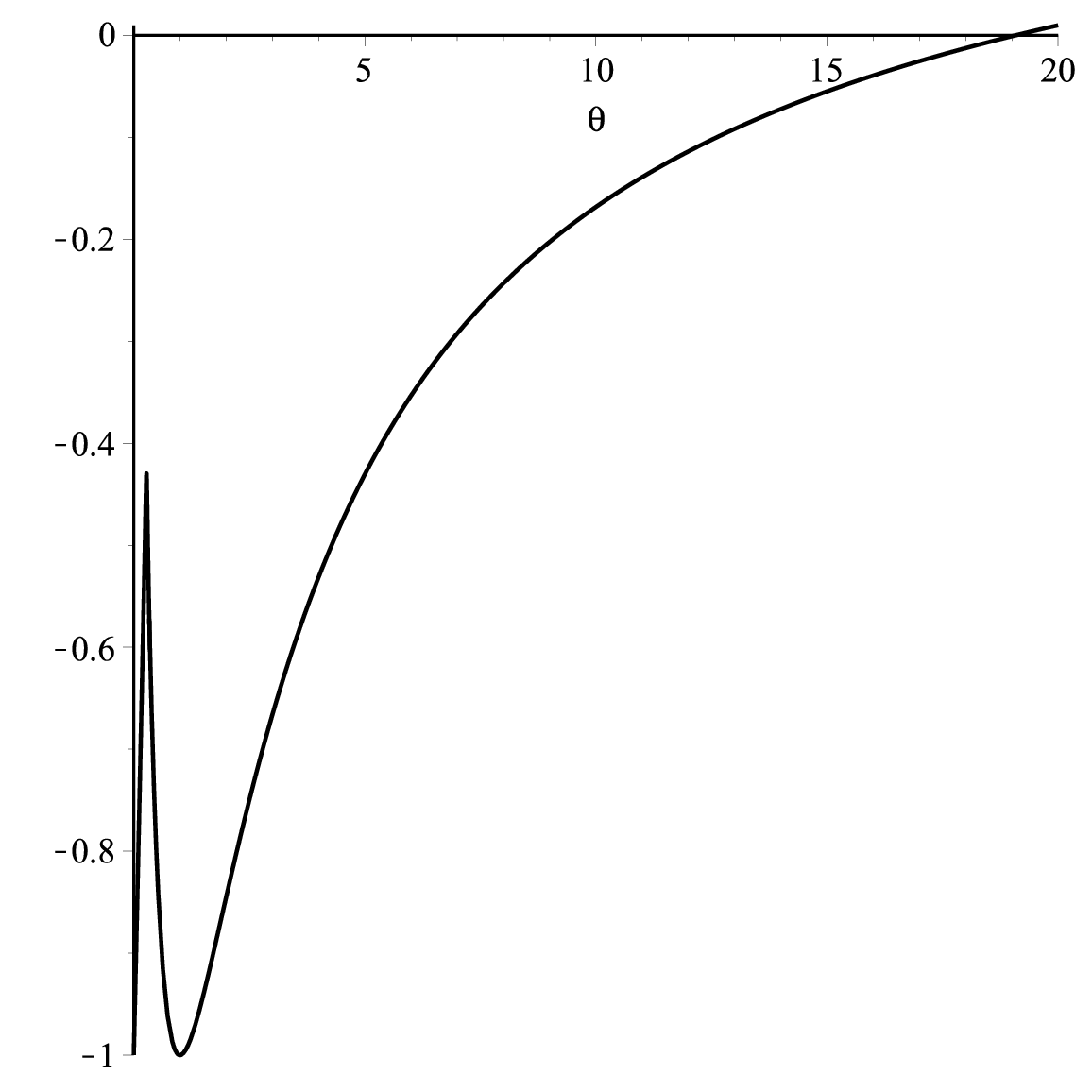} \ \ \ \ \
		\includegraphics[width=7cm]{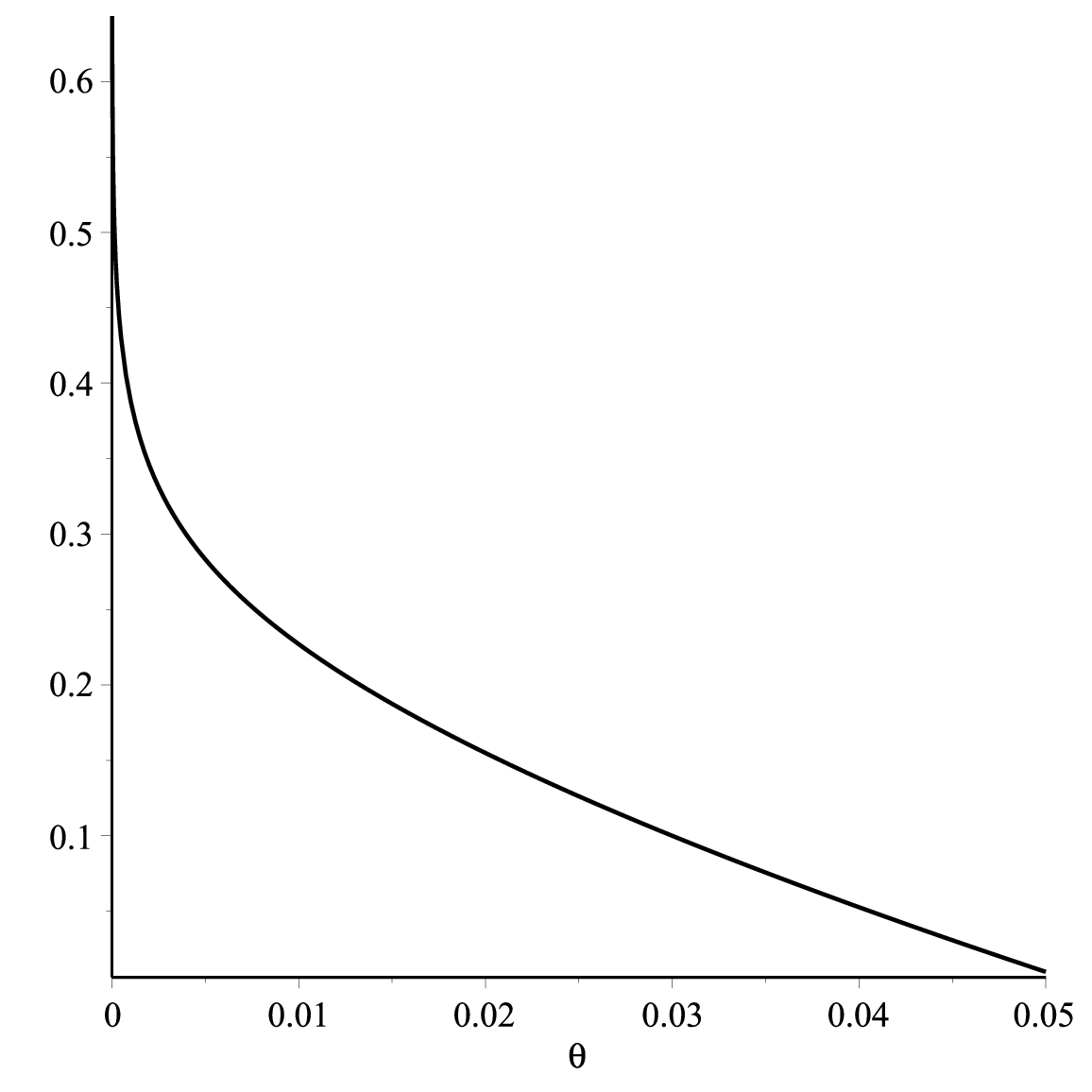}\ \ \
	\end{center}
	\caption{Graphs of the functions $U_1(\theta,0.1)$ for $\theta\in(0, 20)$ (left) and $U_1(1/\theta,0.1)$ for $\theta\in(0, 0.05)$ (right).}\label{e1}
\end{figure}

\medskip
{\bf Subcase: $p=0.1$ and for solutions $y_i$, $i=2,3$:}  For solutions $y_2$ and $y_3$ of~\eqref{y3} from~\eqref{kac},  by computer analysis, we get  
\begin{equation}\label{k2}
	\kappa(1,y_i)={|1-\theta^{2^{0.1}}|\over Z_1}, \ \ i=2,3.
\end{equation}
Denote
$$U_i(\theta, p)=2{1-\theta^{2^p}\over 1+\theta^{2^p}}\kappa(1,y_i)-1.$$
\begin{figure}[h]
	\begin{center}
		\includegraphics[width=7cm]{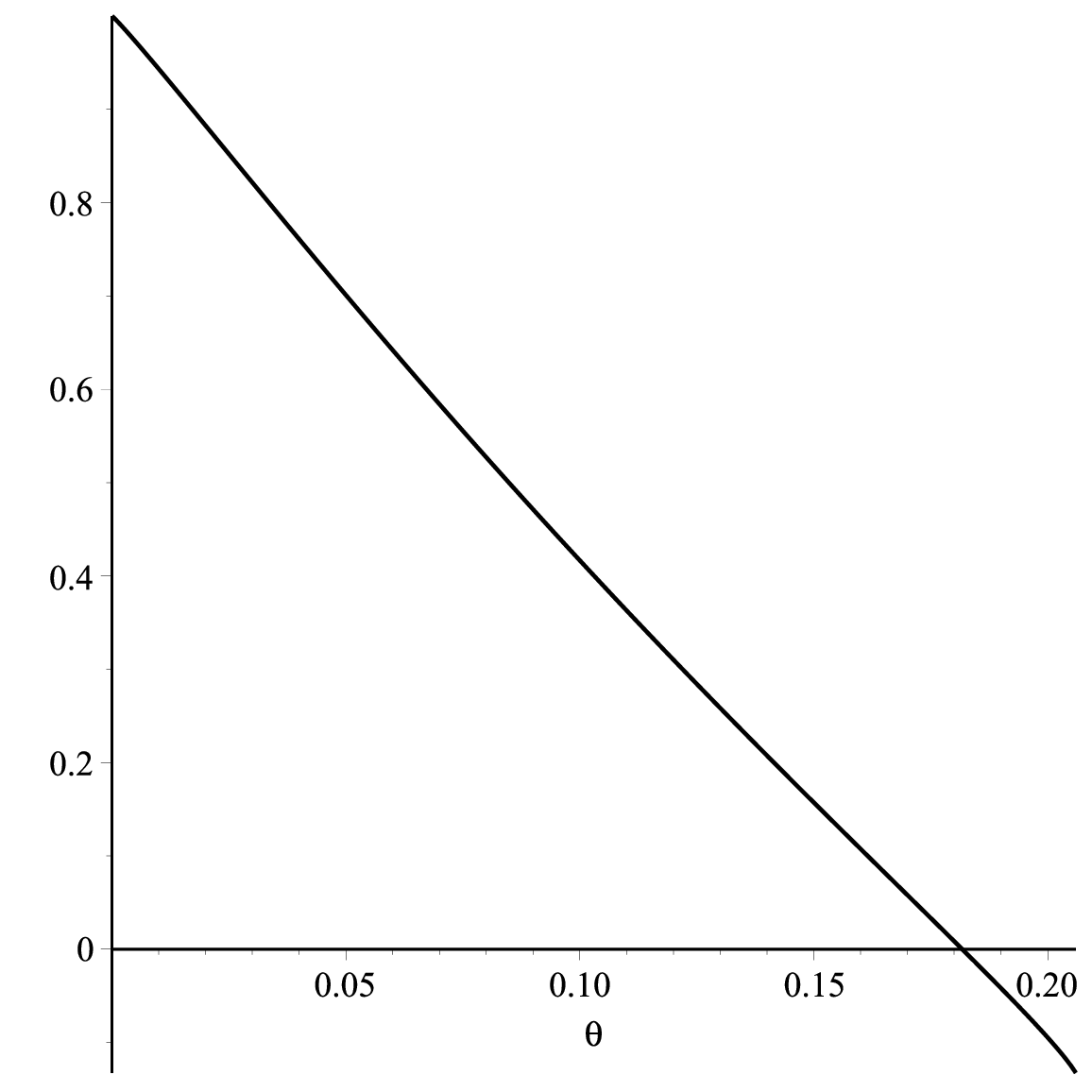} \ \ \ \ \
		\includegraphics[width=7cm]{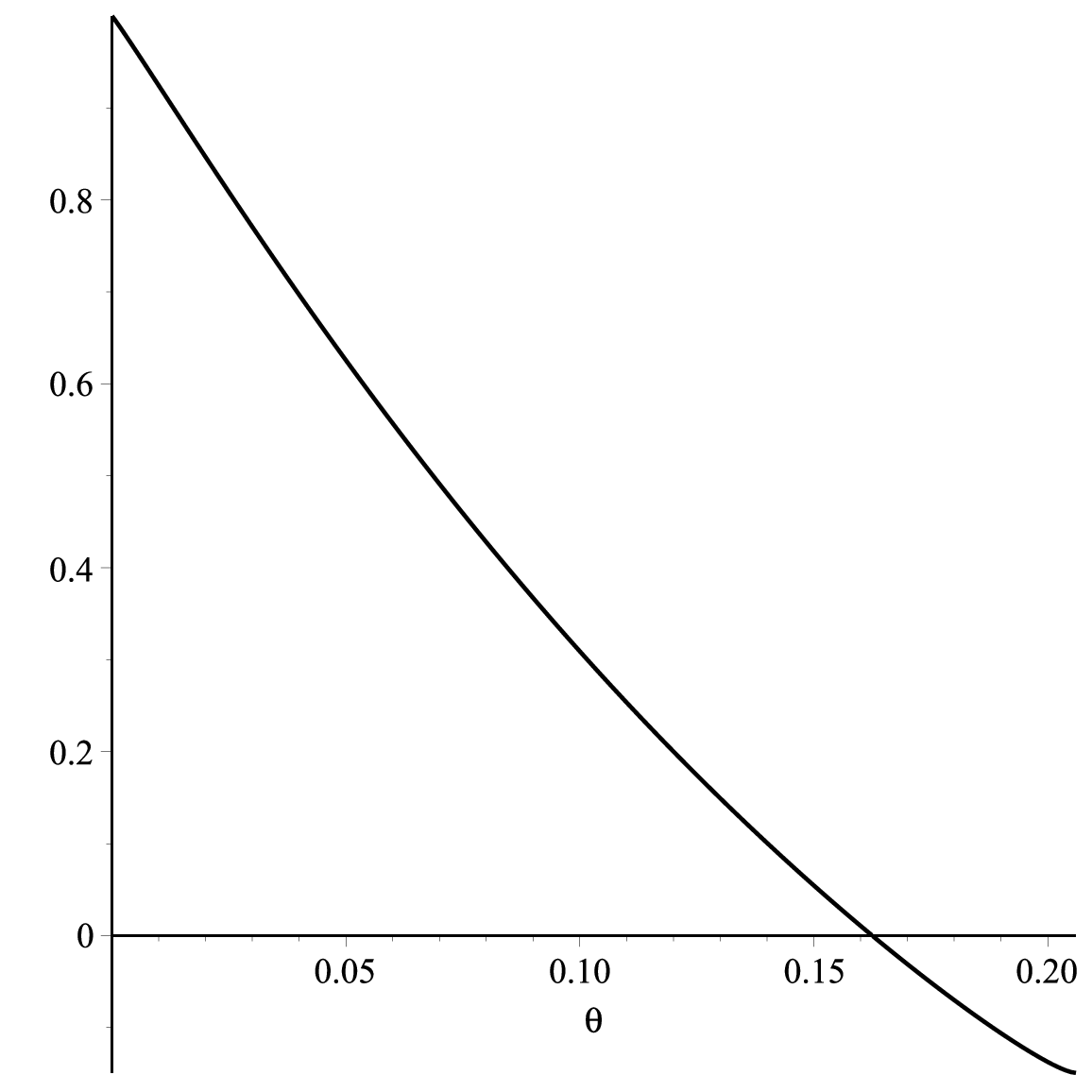}\ \ \
	\end{center}
	\caption{Graphs of the functions $U_2(\theta,0.1)$ (left) and $U_3(\theta, 0.1)$ (right) for $\theta\in(0, \theta_2^*)$. }\label{e2}
\end{figure}
\begin{re} Fig.~\ref{e2} shows that the extremality condition for solutions $(1, y_i(\theta,0.1))$, $i=2,3$, where they exist (i.e., when  $\theta\in (0, \theta_2^*)$), holds for:

- $y_2$ if $\theta\in (\bar \theta_2, \theta_2^*)$, where $\bar\theta_2\approx 0.1817$ and

- $y_3$ if $\theta\in (\bar \theta_3, \theta_2^*)$, where $\bar\theta_3\approx 0.1625$.
\end{re}

{\bf Subcase: $p=10$.} Consider the case $y_1(\theta,10)$, then our computer analysis shows the following:

\begin{re} For the solution $y_1(\theta, 10)$ the extremality condition is satisfied (see Fig.~\ref{e10}) if $\theta\in (0,1)$ and does not hold if $\theta>1$.
\end{re}
\begin{figure}[h!]
	\begin{center}
		\includegraphics[width=7cm]{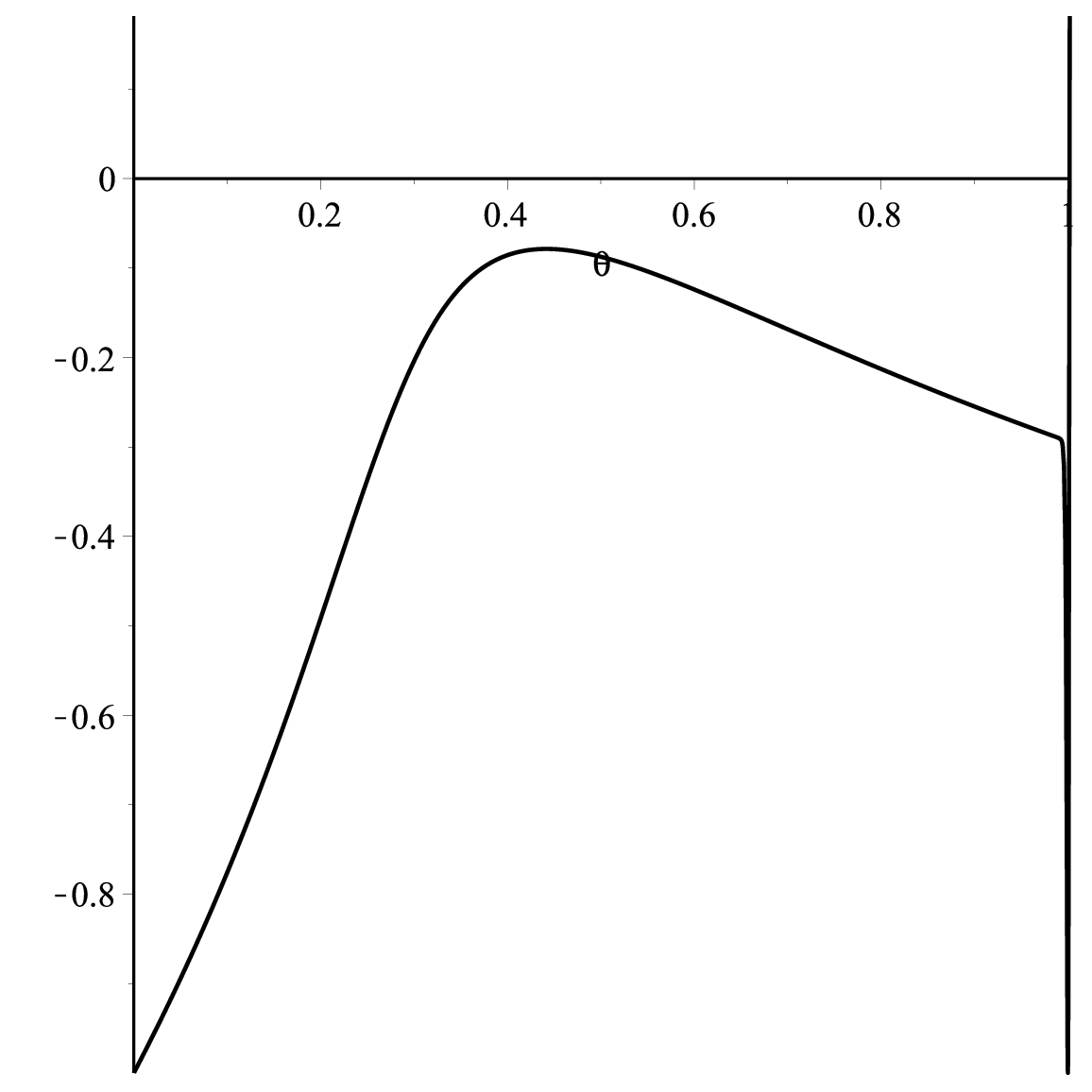} \ \ \ \ \
		\includegraphics[width=7cm]{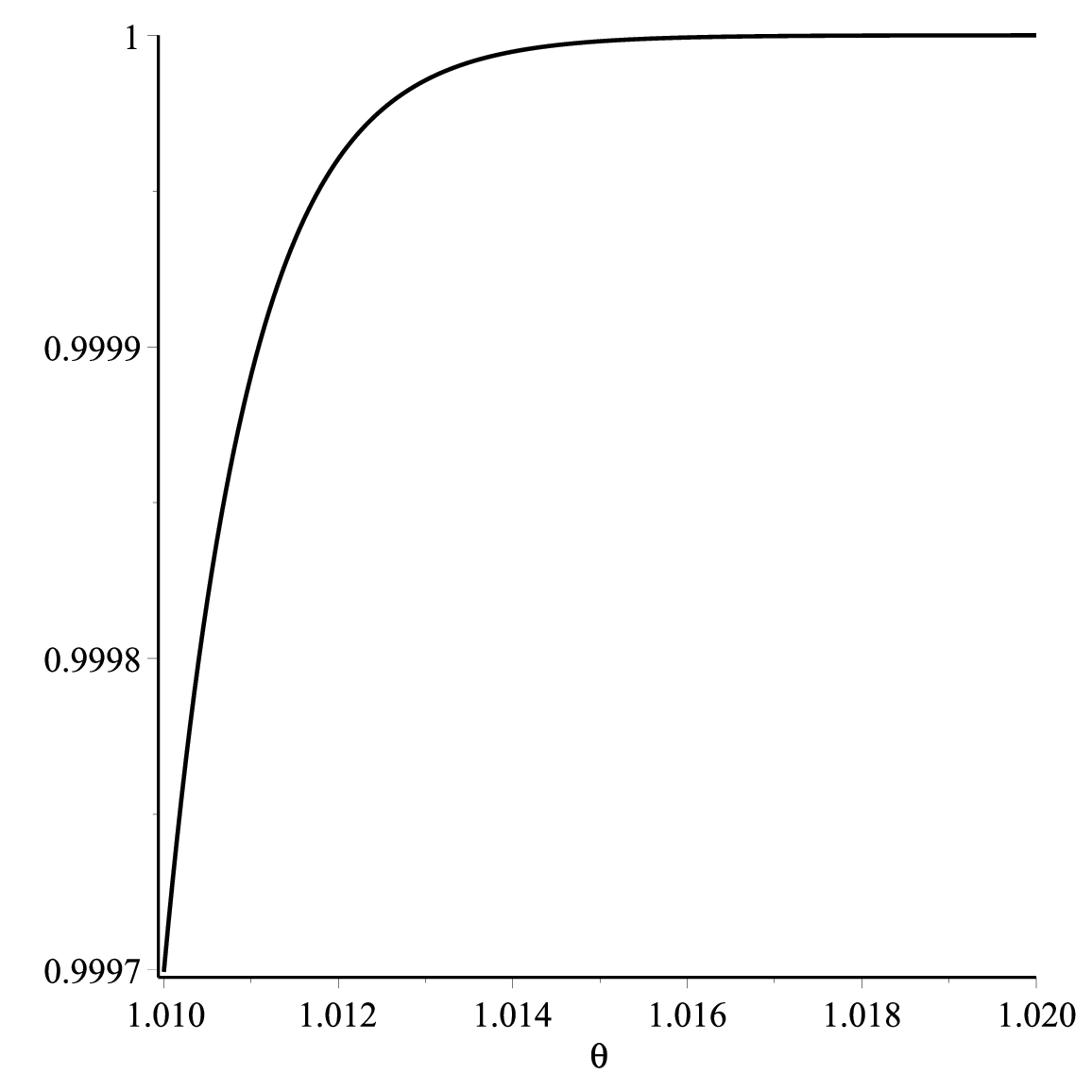}\ \ \
	\end{center}
	\caption{Graph of the function $U_1(\theta,10)$ for $\theta\in(0,1)$ (left) and $\theta>1$ (right). }\label{e10}
\end{figure}
\begin{rk}  In \cite{KR-sos}, it was demonstrated that for $p=1$, the measure $\mu_1(\theta,1)$ associated with $y_1(\theta,1)$ is extreme when $\theta <\approx 2.655$ and non-extreme if $\theta >\approx 2.87$. For $p=0.1$, we have established that $\mu_1(\theta,0.1)$ is extreme when $\theta <\approx 19$ and non-extreme when $\theta>\approx 1523$. However, the unexpected finding emerged in the case of $p=10$, where the critical value distinguishing between extremality and non-extremality is exactly 1. This critical value aligns with the boundary between the ferromagnetic and anti-ferromagnetic cases. 
\end{rk}
\begin{re} For the solutions $y_i(\theta, 10)$, $i=2,3$ Fig.~\ref{e20} shows that the extremality condition is never satisfied. 
\end{re}
\begin{figure}[h!]
	\begin{center}
		\includegraphics[width=7cm]{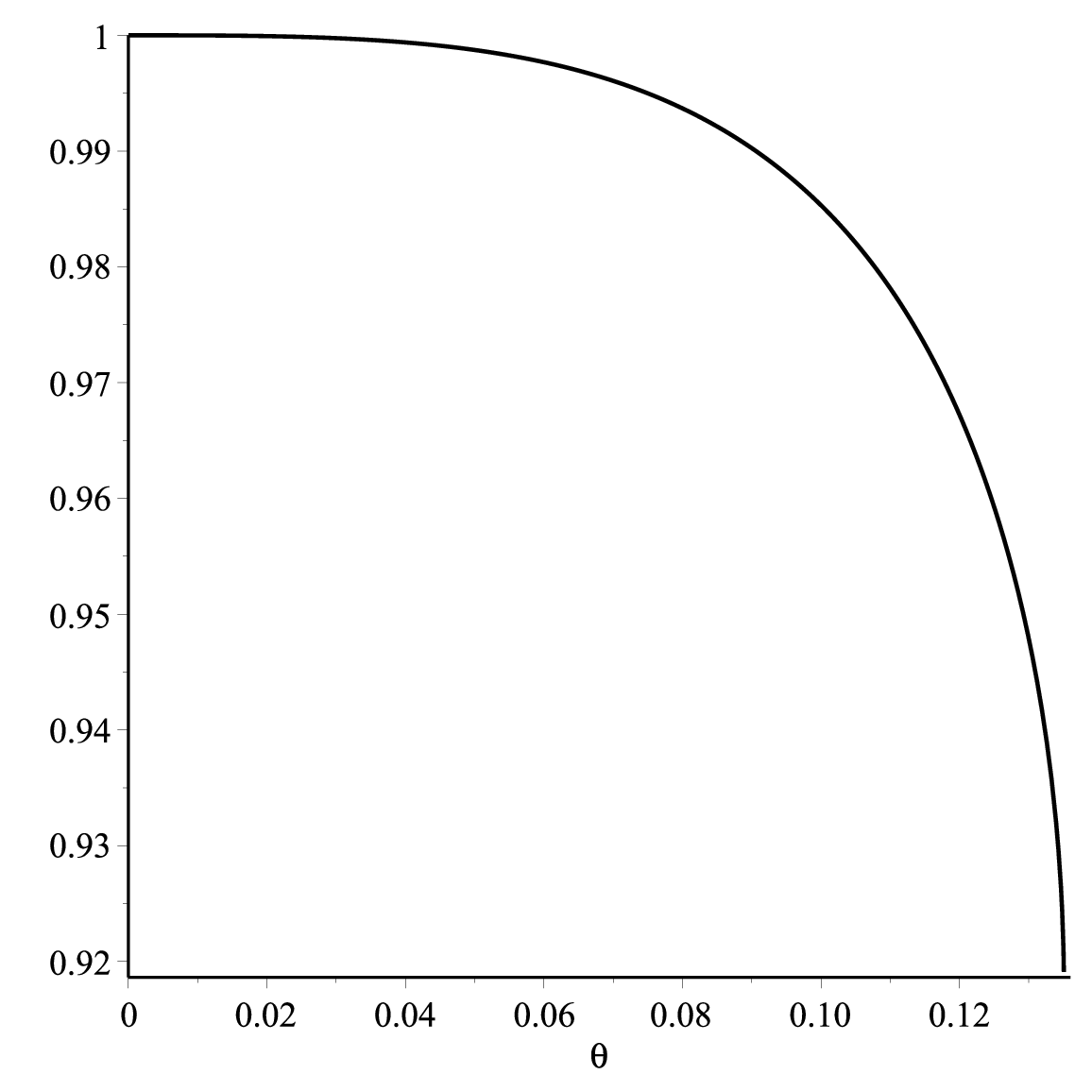} \ \ \ \ \
		\includegraphics[width=7cm]{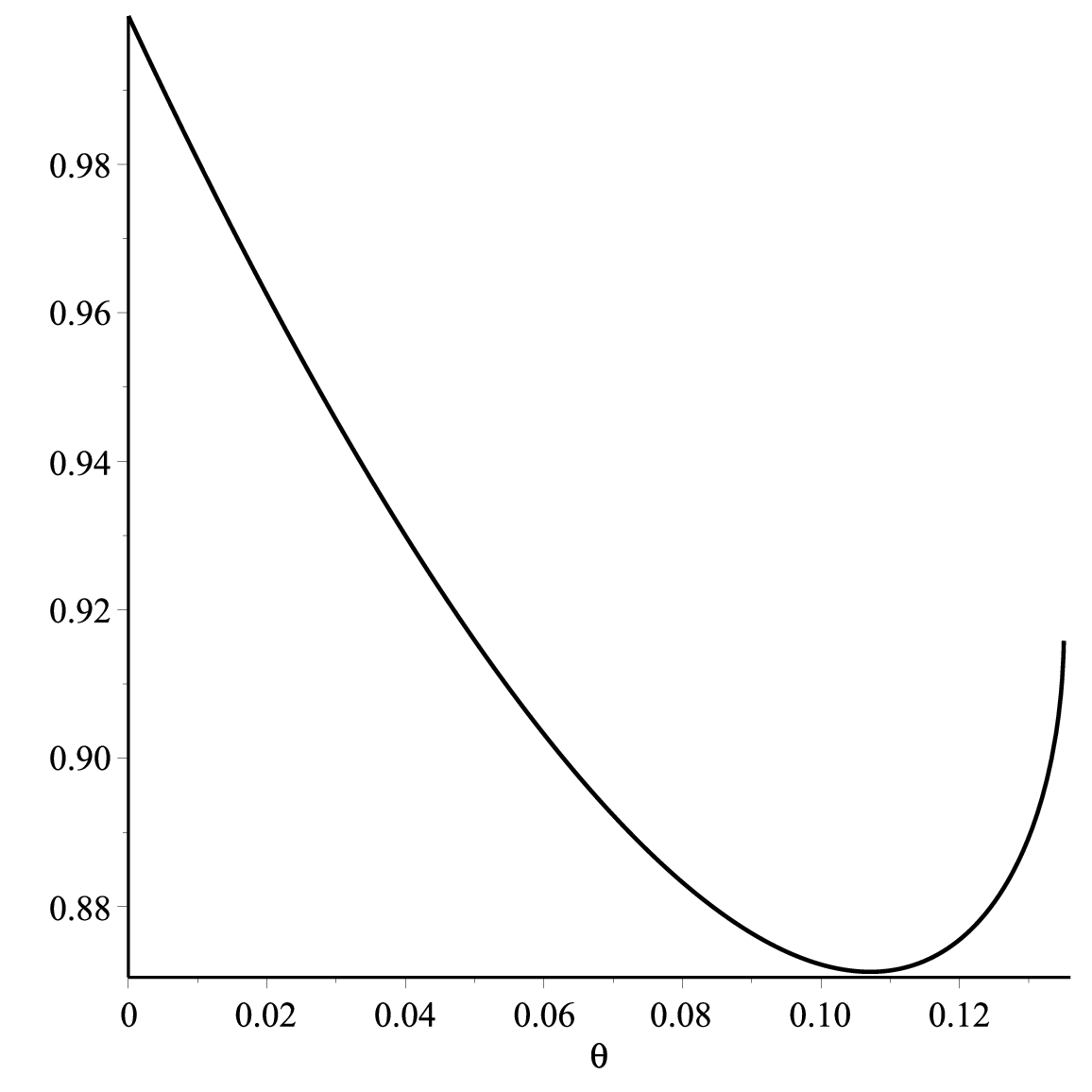}\ \ \
	\end{center}
	\caption{Graphs of the functions $U_2(\theta,10)$ (left) and $U_3(\theta, 10)$ (right) for $\theta\in(0, \theta_2')$. }\label{e20}
\end{figure}

Recall $\mu_i=\mu_i(\theta, p)$ is the SGM corresponding to solution $v_i$. Let us summarize Results 1--8 in the following theorem.
 \begin{thm}\label{tk}
The following holds true for the binary tree.
 \begin{itemize}
 	\item[1.] If $p=0.1$ then
 	\begin{itemize}
	\item[1.a)] There are values $\theta_1^*$ $(\approx 19.08)$ and $\tilde\theta_1$ $(\approx 1523.4)$ such that the measure $\mu_1$ is extreme if $0<\theta<\theta_1^*$ and is non-extreme if $\theta>\tilde\theta_1$.
	\item[1.b)] There are values $\hat\theta_2$ $(\approx 0.175)$ and $\bar\theta_2$ $(\approx 0.1817)$ such that the measure\footnote{Note that $\mu_2$ and $\mu_3$ exist for $\theta\in (0, \theta_2^*)$, where $\theta_2^*\approx 0.206$.} $\mu_2$ is non-extreme if $0<\theta<\hat\theta_2$ and is extreme if $\theta\in (\bar\theta_2, \theta_2^*)$.
	\item[1.c)] There are values $\hat\theta_3$ $(\approx 0.139)$ and $\bar\theta_3$ $(\approx 0.1625)$ such that the measure $\mu_3$ is non-extreme if $0<\theta<\hat\theta_3$ and is extreme if $\theta\in (\bar\theta_3, \theta_2^*)$.
	\end{itemize}
\item[2.] If $p=10$ then
\begin{itemize}
\item[2.a)] The measure $\mu_1$ is extreme if $0<\theta\leq 1$ and is non-extreme if $\theta>1$.
	 \item[2.b)] The measures $\mu_2$ and $\mu_3$ are non-extreme (where they exist).
	\end{itemize}
	\end{itemize}
  \end{thm}

\begin{rk} There is a big gap between the established regions of non-extremality and extremality for the measure $\mu_1$ at $p = 0.1$ (case 1.a) in Theorem \ref{tk}). We conjecture that in the majority of this gap, $\mu_1$ is non-extreme.	
\end{rk}

Let us finally also present a criterion for extremality for the remaining solutions.
 
\medskip

{\bf Case $x\ne 1$.} 
Now we compute $\kappa$ for $(x_i,y_i)$, $i=4,5,6,7$. Recall that all of them exist only for $\theta<1$, therefore, from the system~\eqref{rs3.2a}, \eqref{rs3.2b} we get
the following inequalities
\begin{align*}
y-\theta &={(1-\theta^2)y^2+\theta(1-\theta^{2^p})\over Z}>0,\\
1-\theta^{2^p}x&={(1-\theta^{2^p})\theta y^2+(1-\theta^{2^{p+1}})\over Z}>0,\\
x-\theta^{2^p}&={(1-\theta^{2^p})((\theta^{2^p}+1)x^2+\theta y^2)\over Z}>0,\\
y-\theta&={(1-\theta^{2^p})\theta x^2+(1-\theta^2)y^2)\over Z}>0.
\end{align*}
These inequalities are useful in order to omit the absolute value of the corresponding difference, but still the form of 
$\kappa(x,y)$ remains bulky. Recall that, in order to check the extremality of a given TISGM $\mu_i$ we need to verify that
$2\kappa\gamma<1$. Using the above mentioned bound of $\gamma$ and Formula~\eqref{kac}, it suffices to check
$$2\kappa(x_i,y_i)\gamma(x_i,y_i)\leq  2{1-\theta^{2^p}\over 1+\theta^{2^p}}\kappa(x_i,y_i)<1.$$
Denote
$$U_i(\theta, p)=2{1-\theta^{2^p}\over 1+\theta^{2^p}}\kappa(x_i,y_i)-1\qquad\text{ and}\qquad 
\mathbb E_i=\{(\theta, p): U_i(\theta, p)<0\}.$$
Thus we obtained
the following proposition.

\begin{pro}\label{te}
	If $(\theta, p)\in \mathbb E_i$ is such that $\mu_i$ exists then  $\mu_i$ is extreme.
\end{pro}
\begin{rk}
	In Propositions~\ref{tne} and~\ref{te}, we were unable to explicitly provide the regions of $(\theta, p)$ for (non-)extremality due to the complex nature of the solutions $(x_i, y_i)$. However, our results may still provide the groundwork for future numerical studies of these regions. In this manuscript we prototyped this analysis for the cases where $p=0.1$, $p=10$ and $x=1$.  
\end{rk}

\section*{Data availability statements}
The datasets generated during and/or analysed during the current study are available from the corresponding author (U.A.Rozikov) on reasonable request.

\section*{Conflicts of interest} The authors declare no conflicts of interest.

 \section*{Acknowledgements}

 B.~Jahnel is supported by the Leibniz Association within the Leibniz Junior Research Group on {\em Probabilistic Methods for Dynamic Communication Networks} as part of the Leibniz Competition (grant no.~J105/2020).
 U.~Rozikov thanks the Weierstrass Institute for Applied Analysis and Stochastics, Berlin, Germany for support of his visit.  His  work was partially supported through a grant from the IMU--CDC and the fundamental project (grant no.~F--FA--2021--425) of The Ministry of Innovative Development of the Republic of Uzbekistan.
 
 We appreciate the referee's valuable comments, which have enhanced the readability of the paper.

\end{document}